\documentclass[10pt]{article}
\usepackage{amssymb}
\usepackage{amsmath}
\usepackage[all]{xy}
\usepackage{amsmath}
\usepackage{amsfonts}
\usepackage{amssymb}
\usepackage{amscd}
\usepackage{amsthm}
\usepackage{latexsym}
\usepackage{amsbsy}
\usepackage{graphicx}
\usepackage{float}
\usepackage[toc,page]{appendix}
\usepackage{caption}
\usepackage{subcaption}
\usepackage[utf8]{inputenc}
\usepackage{subcaption}

\usepackage{hyperref}

\newtheorem{theorem}{Theorem}[section]
\newtheorem*{theorem*}{Theorem}
\newtheorem{proposition}[theorem]{Proposition}
\newtheorem*{proposition*}{Proposition}
\newtheorem{lemma}[theorem]{Lemma}

\newtheorem{remark}{Remark}[theorem]
\newtheorem{definition}[theorem]{Definition}
\newtheorem*{notation}{Notation}
\newtheorem{assumption}[theorem]{Assumption}
\newtheorem*{example*}{Example}

\DeclareMathOperator{\tr}{Tr}
\newcommand{\C}{\mathbb{C}}
\DeclareMathOperator{\Cl}{Cl}

\renewcommand{\epsilon}{\varepsilon}
\newcommand{\N}{\mathbb{N}}
\DeclareMathOperator{\Tr}{Tr}
\newcommand{\R}{\mathbb{R}}

\newcommand{\weakto}{\xrightarrow{\text{weak}}}
\newcommand{\into}{\hookrightarrow}
\renewcommand{\limsup}{\overline{\lim}}
\renewcommand{\liminf}{\underline{\lim}}
\DeclareMathOperator{\PV}{P.V.}
\DeclareMathOperator{\supp}{supp}

\DeclareMathOperator{\spec}{spec}
\DeclareMathOperator{\vol}{vol}
\newcommand{\Sym}{\mathfrak{S}}
\newcommand{\Per}{\mathfrak{S}'}

\author{Masoud Khalkhali, Nathan Pagliaroli, and Luuk S. Verhoeven\\
	Department of Mathematics,  Western University\\
	London, Ontario, Canada\footnote{\emph{Email addresses}:  masoud@uwo.ca, npagliar@uwo.ca,  and luuksv@gmail.com }}
\title{Large $N$ limit of fuzzy geometries coupled to fermions}

\begin{document}

\maketitle

\begin{abstract}
  In this paper we present an analysis of the large $N$ limit of a family of quartic Dirac ensembles based on $(0, 1)$ fuzzy geometries that are coupled to fermions.
	These Dirac ensembles are examples of single-matrix, multi-trace matrix ensembles. Additionally, they serve as examples of integer-valued $\beta$-ensembles. Convergence of the spectral density in the large $N$ limit for a large class of such matrix ensembles is proven, improving on existing results.
	The main results of this paper are the addition of the fermionic contribution in the matrix ensemble and the investigation of spectral estimators for finite dimensional spectral triples.
\end{abstract}

\tableofcontents

\section{Introduction}

In \cite{barrett2015matrix},  John Barrett proposed a class of toy models of Euclidean quantum
gravity built around fuzzy geometries. These models were  studied numerically in \cite{barrett2016monte,glaser2017scaling,perez2021multimatrix,barrett2019spectral,hessam2022bootstrapping,d2022numerical} and then were analytically studied in \cite{khalkhali2020phase,khalkhali2022spectral,azarfar2019random,perez2022computing,khalkhali2023coloured,hessam2022double}. For a recent survey see \cite{hessam2022noncommutative}. The core idea is to replace Riemannian metrics by their noncommutative analogue, which according to Connes' distance formula \cite{connes1994noncommutative} can be defined using a Dirac operator. A general setting for such an approach is the theory of finite spectral triples \cite{krajewski1998classification,connes1995noncommutative,barrett2015matrix}. We refer to a probability distribution on a set of finite spectral triples as a Dirac ensemble. Path integrals over such spectral triples are typically convergent. In particular, we are interested in a class of finite spectral triples called \textit{fuzzy geometries} or fuzzy spectral triples whose path integrals are matrix integrals. Our interest in these models is partly due to the fact that these models act as a bridge between Noncommutative Geometry and Random Matrix Theory. Other such connections can found in \cite{van2022one,wulkenhaar2019quantum,branahl2022scalar,Gesteau2024dhj}

In this paper we study a certain class of fermionic fuzzy Dirac ensembles. These are Dirac ensembles based on fuzzy geometries  that are coupled to fermions in the sense of \cite{barrett2024fermion}. 
In particular we study the type $(0, 1)$-fuzzy geometries $(M_{N}(\mathbb{C}), M_{N}(\mathbb{C}),D)$ where $D = [H,\cdot]$, and $H$ is a Hermitian $N \times N$ matrix.  The Dirac ensemble of interest in this paper is a Hermitian single-matrix multi-tracial ensemble of the form
\[
	Z_{\text{fuzzy}} = \int e^{- g_{4}\tr D^{4} - g_{2}\tr D^{2}} dD = \int_{\mathcal{H}_N} e^{-\Tr\left(g_4\left(H \otimes 1 - 1 \otimes H\right)^4 - g_2\left(H \otimes 1 - 1 \otimes H\right)^2\right)} dH.
\]
The addition of a  fermion introduces a fermionic contribution to the action of the form
\begin{align*}
	Z_{D,\psi} &= \int  \exp\left(-\Tr\left(g_4 D^4 + g_2 D^2\right) - \langle \psi, D\psi \rangle\right) dD d\psi.
\end{align*}
 Using Weyl's integration formula, this integral can be reduced to an $N$-dimensional integral, and the fermionic contribution can be absorbed into the action as a term similar to $\sum_{i,j} \log (\lambda_i - \lambda_j)$, where the sum is over the eigenvalues of $H$. As we will see in section \ref{sec:the_Dirac_ensembles}, to get non-trivial results, this action needs to be regularized by adding a mass term.

It is well-known in the literature \cite{johansson1998fluctuations,deift1999orthogonal} that under reasonable conditions on single-matrix single-trace ensembles there exists a continuous, compactly supported probability measure $\mu_E$ such that, for bounded functions $f$,
\[
	\lim_{N\to\infty} \langle \frac{1}{N} \tr\left(f(H)\right) \rangle = \int_{\R} f(x) d\mu_E.
\]
This $\mu_E$ is called the large $N$ spectral density, or spectral density for brevity, of the ensemble. This $\mu_E$ is also the measure minimizing an energy function corresponding to the ensemble, and in this guise it is often called the equilibrium measure.
We improve on these results to also cover the multi-tracial, non-polynomial model we encounter as our fermionic fuzzy Dirac ensemble.

With these results on random matrix ensembles established we investigate the fermionic fuzzy Dirac ensembles.
Because the large $N$ limit of the spectral density can only be used for normalized traces we do not obtain access to the limit of the heat kernel, but only the limit of the normalized heat kernel
\[
	k_{D^2}(t) := \frac{1}{\dim D} \tr\left(e^{tD^2}\right).
\]
This prevents us from obtaining a dimension or volume from the spectral density with the usual heat kernel techniques.
Instead we use the \emph{spectral dimension} and \emph{spectral variance} as proxies for the dimension. These quantities that measure the dimension of a spectrum were developed in \cite{barrett2019spectral} for the context of finite fuzzy geometries such as those appearing in the Monte-Carlo simulations of \cite{barrett2016monte}.
These results can be found in Section \ref{sec:spectral_properties}.

This paper is organized as follows. In Section \ref{subsec:fuzzy_geometries} we recall the details of fuzzy geometries, which we then use to define the $(0, 1)$-fuzzy Dirac ensemble in Section \ref{subsec:quartic_assoc_matrix_model}. Then in Section \ref{subsec:fermionic_model} we  incorporate the fermionic action.

Section \ref{sec:the_spectral_density} is dedicated to the existence of the spectral density for multi-tracial matrix ensembles. The key ingredients and main results are discussed in Section \ref{subsec:existence_of_eq_measure}, with the technical details in Section \ref{subsec:proof_that_eq_measure_is_spec_density}. In Section \ref{subsec:computation_of_eq_measure} we briefly discuss the computation of the equilibrium measure.

Finally in Section \ref{subsec:estimators_massless} we discuss properties such as the spectral dimension and the spectral phase transition of the fermionic enemble for massless fermions.
Then in Section \ref{subsec:estimators_non_zero_mass} we examine the influence of the mass of the fermion on these various properties.

\subsection*{Acknowledgements}

	We would like to thank John Barrett and L Glaser for useful discussions regarding  the various aspects of the fermionic model. We would like to thank the Natural Sciences and Engineering  Research Council of Canada (NSERC) for financial support.

\section{The Dirac ensembles}
\label{sec:the_Dirac_ensembles}

	We start this section by giving a very brief and targeted introduction to spectral triples and matrix ensembles. For much more complete introductions to these concepts, see for example \cite{gracia2013elements,van2015noncommutative,eynardRandomMatrices,deift1999orthogonal}.

	Spectral triples \cite{connes1995noncommutative,connes2019noncommutative} are noncommutative analogues of Riemannian spin manifolds.   In their simplest form they consist of a triple
	\[
		\left(A, H, D\right)
	\]
	where $A$ is some $*-$algebra, $H$ is a Hilbert space equipped with a representation of $A$, and $D$ is a, usually unbounded, operator on $H$.
	This triple is further required to satisfy several analytical requirements, but if $H$ is finite dimensional as it will be throughout this paper, these requirements are automatically satisfied.
	The spectral triple corresponding to a Riemannian manifold $M$ with spinor bundle $S \to M$ is given by
	\[
		\left(C^\infty(M), L^2(S \to M), D_M\right),
	\]

	where $D_{M}$ is the Dirac operator. Often further structure on the spectral triples is desired. In our case, this is a real structure and a grading, which 
 turns the spectral triple into a quintuple
	\[
		\left(A, H, D; J, \Gamma\right),
	\]
	where the real structure $J$ is an anti-unitary map with $J^2 = \pm 1$ and the grading $\Gamma$ is a linear map with $\Gamma^2 = 1$.
	These structures further require $\Gamma D = -D \Gamma$, $DJ = \pm JD$ and $\Gamma J = \pm J \Gamma$.
	The three optional signs appearing in these conditions together determine the $KO$-dimension of the spectral triple by Table \ref{tab:KO_dimension}.
	Additionally we demand the order-zero condition
	\[
		\left[a, J^{-1}b^*J\right] = 0,
	\]
	and order-one condition
	\[
		\left[\left[D, a\right], J^{-1}b^*J\right] = 0.
	\]
	These conditions imply that $J$ can be used to construct a right representation of $A$ on $H$ that commutes with the left representation, and that the action of ``derivatives'' $[D, a]$ is on the left.

	\begin{table}[h]
		\centering
		\begin{tabular}{r | c c c c c c c c}
			$s \mod 8$ 		& 0 & 1 & 2 & 3 & 4 & 5 & 6 & 7 \\
			\hline
			$\epsilon$		& $+$ & $+$ & $-$ & $-$ & $-$ & $-$ & $+$ & $+$ \\
			$\epsilon'$		& $+$ & $-$ & $+$ & $+$ & $+$ & $-$ & $+$ & $+$ \\
			$\epsilon''$	& $+$ & $+$ & $-$ & $+$ & $+$ & $+$ & $-$ & $+$
		\end{tabular}
		\caption{The usual choices of the signs associated to a $KO$-dimension $s$ in noncommutative geometry \cite{connes1995noncommutative,connes2019noncommutative}. The signs are determined by $J^2 = \epsilon$, $JD = \epsilon'DJ$ and $J\Gamma = \epsilon''\Gamma J$. If the $KO$-dimension is odd the grading $\Gamma$ is set to be trivial.}
		\label{tab:KO_dimension}
	\end{table}

	Given the algebra, Hilbert space, real structure and grading of a spectral triple there is a real vector space of operators that can play the role of a Dirac operator completing the spectral triple.
	By a Dirac ensemble we mean a probability distribution on a subset of the space of possible Dirac operators for a given algebra, Hilbert space, and possibly additional structures.
	In this paper we will consider a specific Dirac ensemble, based on a specific spectral triple.
	These are chosen such that the corresponding matrix ensemble has an equilibrium measure that can be found by the Coulomb gas technique.
	In the next section we introduce a wider context for fuzzy geometries. 	We are hopeful that in the future wider classes of Dirac ensembles can be opened up to analytic  investigation. This shall be elaborated on further in the discussion of Section \ref{sec:discussion}.
	
	\subsection{Fuzzy geometries}
	\label{subsec:fuzzy_geometries}

		The terminology ``fuzzy geometry'' was introduced in \cite{barrett2015matrix}. It refers to the class of spectral triples in which the algebra is $M_N(\C)$, for some $N$, and the Hilbert space consists of an irreducible Clifford module tensored with $M_N(\C)$, representing a trivial spinor bundle over the function algebra.
		The motivation for the name comes from the examples of the fuzzy sphere and fuzzy torus \cite{dandrea2013metric,barrett2019torus,madore1992fuzzy}, which are examples of such spectral triples.

		\begin{definition}
			A \emph{fuzzy geometry} of signature $(p, q)$ is a real spectral triple of $KO$-dimension $q-p$ of the form
			$$
				\left(M_N(\C), M_N(\C) \otimes V, D; J = 1 \otimes J_V\right)
			$$
			where $V$ is an irreducible Clifford module for the Clifford algebra $\Cl_{p, q}$ with real structure $J_V$.
			If $p + q$ is even, the fuzzy geometry is an even real spectral triple with grading $\Gamma = 1 \otimes \Gamma_V$.
		\end{definition}

		The real structure, as implied by the order one condition on real spectral triples, restricts the form of the Dirac operator as follows.

		\begin{theorem}[Barrett \cite{barrett2015matrix}]
			\label{thm:Barrett_fuzzy_Diracs}
			The Dirac operator of a signature $(p, q)$ fuzzy geometry can be written as
			$$
				D(A \otimes \psi) = \sum_{I} (K_I A + \epsilon' A K_I^*) \otimes \gamma_I \psi.
			$$
			Here the sum ranges over certain subsets $I \subset \{1, \ldots, p+q\}$ and $\gamma_I = \prod_{i \in I} \gamma_i$. 
			The $K_I$ are matrices that are either skew- or self-adjoint but otherwise arbitrary.
		\end{theorem}

		For particular signatures $(p, q)$, there are more details available about which subsets $I$ are allowed in the sum and which $K_I$ are skew- and which are self-adjoint.
		To avoid filling the current paper with the various possible cases, we refer for these details to \cite{barrett2015matrix}.
		This theorem implies that a Dirac ensemble of fuzzy geometries yields a \emph{Hermitian matrix ensemble} by parametrizing $D$ in terms of the $K_I$ and $i K_I$ where appropriate.

		In this the paper we want to highlight the fuzzy geometries of signature $(0, 1)$. These take the form
		$$
			\left(M_N(\C), M_N(\C), D = [H, \cdot]; J = \cdot^*\right)
		$$
		where $H$ is an arbitrary self-adjoint matrix.
		The space $\mathcal{D}_N$ of all possible Dirac operators can thus be described as $\mathcal{D}_N = \left\{[H, \cdot]\right\}_{H \in \mathcal{H}_N}$.
		Dirac ensembles based on these fuzzy geometries thus yield a \emph{single matrix} Hermitian matrix ensemble. For such random matrix ensembles various tools exist  to investigate the average spectrum of $H$, and thus also the average spectrum of $D$.

	\subsection{The quartic \texorpdfstring{$(0, 1)$}{(0, 1)}-fuzzy Dirac ensemble}
	\label{subsec:quartic_assoc_matrix_model}

		Before introducing our Dirac ensemble of interest we give a brief introduction to matrix ensembles. A proper introduction can, for example, be found in \cite{deift1999orthogonal,eynardRandomMatrices,guionnet2009large}.
		A matrix ensemble consists, in essence, of a collection of matrices together with a probability distribution. One is then interested in the probability distribution of properties of these matrices, such as the smallest eigenvalue or its trace.
		One of the most famous examples of a matrix ensemble is the \emph{Gaussian Unitary Ensemble}, given by the space of $N \times N$ Hermitian matrices $\mathcal{H}_N$ with probability density
		\[
			p(H) = \frac{1}{Z} e^{-\Tr(H^2)} dH,
		\]
		where $dH$ is the Lebesgue measure on $\mathcal{H}_N \cong \R^{N^2}$ (in any basis) and
		\[
			Z = \int_{\mathcal{H}_N} e^{-\Tr(H^2)} dH
		\]
		is the normalization factor computed with the same Lebesgue measure.

		In general, matrix ensembles can have multiple matrices as dynamical variables, such as in the $ABAB$-model \cite{kazakov1999two}
		\[
			p(A, B) = \frac{1}{Z} e^{- g_4\left(\Tr(A^4) + \Tr(B^4)\right) - g_{ABAB} \Tr(ABAB) } dA dB
		\]
		where $A, B \in \mathcal{H}_N$, $Z$ is again the normalization factor and the coefficients $g_4$ and $g_{ABAB}$ are model parameters called \emph{coupling constants}.
		We restrict ourselves to single-matrix ensembles since that is the setting in which the spectral density is available.

		By a Dirac ensemble we then mean a collection of Dirac operators completing some otherwise fixed spectral triple, together with a probability distribution on these Dirac operators.
		When considering fuzzy geometries as the spectral triple, such a Dirac ensemble takes the form of a matrix ensemble, by Theorem \ref{thm:Barrett_fuzzy_Diracs} from \cite{barrett2015matrix}.
		Since matrix ensembles have a significantly longer history than Dirac ensembles, this gives us access to many tools.

		We define the \emph{quartic $(0, 1)$-fuzzy Dirac ensemble} to have as space of Dirac operators $\mathcal{D}_N = \left\{[H, \cdot]\right\}_{H \in \mathcal{H}_N}$ with probability density
		$$
			p(D) = \frac{1}{Z_D} \exp\left(-\Tr\left(g_4 D^4 + g_2 D^2\right)\right) dD.
		$$
		Here $dD$ is the Lebesgue measure on $\mathcal{D} \subset \mathcal{H}_{N^2} \cong \R^{N^2 \times N^2}$ and
		$$
			Z_D = \int_{\mathcal{D}} \exp\left(-\Tr\left(g_4 D^4 + g_2 D^2\right)\right) dD
		$$
		is the normalization factor.
		The fact that $Z$ is defined with respect to the same $dD$ also ensures independence from the exact identification of $\mathcal{D}$ as a subset of the $N^2 \times N^2$ Hermitian matrices.

		There is one hurdle in defining the associated single matrix ensemble, the parametrization $\mathcal{H}_N \ni H \mapsto D = [H, \cdot] \in \mathcal{D}_N$ is not injective.
		Hence the map $\mathcal{H}_N \to \mathcal{D}_N$ is not proper, so that we cannot naively pull back a probability density from $\mathcal{D}_N$ to $\mathcal{H}_N$.
		This can be solved by equipping each fiber $\left\{H = H_0 + tI\right\}_{t \in \R}$ with the probability density $\sqrt{\frac{a}{\pi}}e^{-a\Tr(H)^2}$.
		
		This may seem overly complicated compared to simply restricting the parametrization to traceless Hermitian matrices, which would correspond to putting a Dirac delta as a probability density on the fibers. 
		The reason for our choice is to ensure that the induced density on $\mathcal{H}_N$ is absolutely continuous with respect to the Lebesgue measure on $\mathcal{H}_N$, and that it can be written as a multi-tracial model of fixed order.
		All conclusions are independent of the non-physical parameter $a$.

		Given the preceding discussion, we define the quartic $(0, 1)$-fuzzy Dirac ensemble as the following single-matrix, multi-trace Hermitian matrix ensemble:
		$$
			p(H) = \frac{1}{Z_H} \exp\left(-\Tr\left(g_4 D_H^4 + g_2 D_H^2\right)\right) \exp\left(-a \Tr(H)^2\right) dH
		$$
		$$
			Z_{H} = \int_{\mathcal{H}_N} \exp\left(-\Tr\left(g_4 D_H^4 + g_2 D_H^2\right)\right) \exp\left(-a \Tr(H)^2\right) dH,
		$$
		where $dH$ is the Lebesgue measure on $\mathcal{H}_N \cong \R^{N^2}$.
		We will use spectral density techniques from random matrix theory to investigate the random fuzzy geometry defined by the Dirac operator $D_H$, both in this model and models further incorporating the fermionic action.

	\subsection{The fermionic quartic \texorpdfstring{$(0, 1)$}{(0, 1)}-fuzzy Dirac ensemble}
	\label{subsec:fermionic_model}

		Dirac ensembles based on fuzzy geometries have been studied before \cite{barrett2019spectral,barrett2016monte,khalkhali2020phase,khalkhali2022spectral, perez2022computing,perez2022multimatrix}.
		One of our goals in this paper is to incorporate the fermionic action from the noncommutative geometry approach to particle physics \cite{barrett2024fermion}.
		Our initial goal will be to simply investigate the effect that the presence of a fermionic action has on a Dirac ensemble. In particular, we will not yet consider observables of the fermionic variables.
		This allows us to consider the fermionic action as an additional term in the action.

		As our fermionic space we take the Grassmann algebra generated by the Hilbert space of our spectral triple, $\mathcal{F}_N = \wedge M_N(\C)$.
		Integration over $\mathcal{F}_N$ is therefore Grassmann, or Berezin, integration.
		The spectral action principle \cite{chamseddine1997spectral} tells us that the probability density on the full space $\mathcal{D}_N \times \mathcal{F}_N$, where $\mathcal{F}_N$ is the space of fermions, should be of the form
		$$
			p(D, \psi) = \frac{1}{Z_{D,\psi}} \exp\left(-\Tr\left(g_4 D^4 + g_2 D^2\right) - \langle \psi, D\psi \rangle\right) dD d\psi,
		$$
		if we choose a quartic polynomial as our bosonic action.
		By not considering $\psi$ as an observable of interest we can integrate out $\psi$ to obtain a new probability density on $\mathcal{D}_N$ given by
		\begin{align}
			p(D) & = \frac{1}{Z_{D,\psi}} \exp\left(-\Tr\left(g_4 D^4 + g_2 D^2\right)\right) \int_{\mathcal{F}_N} \exp\left(-\langle \psi, D\psi \rangle\right) d\psi dD, \nonumber \\
			& = \frac{1}{Z_{D,\psi}} \exp\left(-\Tr\left(g_4 D^4 + g_2 D^2\right) + \log\left(F(D)\right)\right) dD,
			\label{eq:fermionic_density_nomass}
		\end{align}
		where
		\begin{equation}
			F(D) = \int_{\mathcal{F}_N} \exp\left(-\langle \psi, D \psi\rangle\right) d\psi.
			\label{eq:base_fermionic_action}
		\end{equation}

		From standard results on Grassmann integration we get $F(D) = \det(D)$. This brings us to a hurdle since $D = [H, \cdot]$ has a nontrivial kernel, so $F(D) \equiv 0$.
		This means that Equation \ref{eq:fermionic_density_nomass} does not define a probability density; in fact, it is ill-defined since $Z_{D,\psi}$ will be 0 as well.
		Together these two zeroes suggest that some form of regularization can be used here. The route we take is the addition of a mass term to the Dirac operator $D$.

		In order to add a mass term to the Dirac operator we need to change the spectral triple. Naively taking $D + m$, for $m \neq 0$, is not a valid Dirac operator on the $(0, 1)$-fuzzy geometry because it violates the requirement $JD = -DJ$, while $D + im$ is not self-adjoint.
		We again take inspiration from the noncommutative geometry approach to particle physics where the full spectral triple is a product of a geometric (commutative) spectral triple and a finite spectral triple containing information about the particles \cite{ChamsedinneConnesMarcolli_ImprovedStandardModel}.

		As finite space we take the spectral triple 
		\begin{equation}
			(\C, \C, m; J = \overline{\cdot})
			\label{eq:finite_spectral_triple}
		\end{equation}
		of $KO$-dimension 7 so that the product with our geometric $(0, 1)$-fuzzy geometry becomes $KO$-dimension 0, which is given by
		$$
			\left(M_N(\C), M_N(\C) \otimes \C^2, D = D_{fuzzy} \otimes \sigma_1 + m \otimes \sigma_2; J, \Gamma = 1 \otimes \sigma_3 \right),
		$$
		where the $\sigma_i$ are Pauli matrices, $\Gamma$ is the grading of the now even spectral triple and $J(A \otimes v) = A^* \otimes \sigma_3 \overline{v}$ \cite{farnsworth2017productsoftriples}.
		We will briefly discuss more complex finite spaces in the discussion at the end of the paper.

		The addition of the mass $m \neq 0$ changes the above computation of the fermionic action.
		A straightforward exercise in linear algebra shows that the spectrum of the Dirac operator of a $(0, 1)$-fuzzy geometry, $D_{fuzzy}$, is given by
		\begin{equation}
			\spec(D_{fuzzy}) = \left\{\lambda_i - \lambda_j \,\middle|\, \lambda_i, \lambda_j \in \spec(H)\right\}
			\label{eq:spectrum_of_fuzzy_D_from_H},
		\end{equation}
		and after the addition of the finite spectral triple we obtain
		\begin{equation}
			\spec(D) = \left\{\pm \sqrt{m^2 + (\lambda_i - \lambda_j)^2} \,\middle|\, \lambda_i, \lambda_j \in \spec(H)\right\}
			\label{eq:spectrum_of_D_with_mass_from_H},		
		\end{equation}
		where each \emph{ordered pair} of eigenvalues of $H$ appears once, leading to a duplication of eigenvalues.
		The duplication of eigenvalues we find is reminiscent of the fermion doubling problem previously encountered in the particle physics models of noncommutative geometry \cite{ChamsedinneConnesMarcolli_ImprovedStandardModel,Barrett_FermionDoubling}, although our context is different.
		This does lead us to consider both the ``default'' fermionic action $\langle \psi, D \psi \rangle$ on the complex Grassmann algebra as well as the ``reduced'' fermionic action $\frac{1}{2}\langle J\psi, D\psi \rangle$ on the real Grassmann algebra.

		The fermionic space $\mathcal{F}_N$ is now the Grassmann algebra generated by $M_N(\C) \otimes \C^2$ and the fermionic contribution to the action is now
		$$
			F(D) = \det(D) \propto \prod_{i, j} \left(m^2 + (\lambda_i - \lambda_j)^2\right)
		$$
		which is indeed non-zero.
		The ``reduced'' fermionic action can also be computed, again using standard results on Grassmann integration, to be
		\begin{equation}
			F(D) = \int_{\mathcal{F}_N} \exp\left(-\frac{1}{2}\langle J\psi, D\psi \rangle\right) d\psi \propto m^{2N} \prod_{i < j} \left(m^2 + (\lambda_i - \lambda_j)^2\right).
			\label{eq:reduced_fermionic_action}
		\end{equation}
		
		Finally, we arrive at the fermionic $(0, 1)$-fuzzy Dirac ensemble. It is a space of Dirac operators
		$$
			\mathcal{D}_N = \left\{[H, \cdot] \otimes \sigma_1 + m \otimes \sigma_2\right\}_{H \in \mathcal{H}_N}
		$$
		equipped with probability density
		$$
			p(D) = \frac{1}{Z_{D,\psi}} \exp\left(-\Tr\left(g_4 D^4 + g_2 D^2\right) + \frac{\beta_2}{4}\sum_{i, j} \log\left(m^2 + (\lambda_i - \lambda_j)^2\right)\right).
		$$
		The sum is over the eigenvalues of $H$ and corresponds to the fermionic action $\log(F(D))$. 
		For the ``default'' fermionic action $\beta_2 = 2$, while for the ``reduced'' fermionic action $\beta_2 = 1$.

		The coupling constant $\beta_2$ plays a very similar role to the Dyson exponent in random matrix theory.
		When considering a random matrix ensemble of self-adjoint matrices there is a repulsive effect between the eigenvalues originating from a factor $\prod_{i < j} |\lambda_i - \lambda_j|^\beta$ appearing in Weyl's integration formula \cite[Ch. 5.3]{deift1999orthogonal}.
		The parameter $\beta$ is called the Dyson exponent and is 1 for real self-adjoint, 2 for complex self-adjoint and 4 for quaternionic self-adjoint matrices.
		The similarity of the eigenvalue repulsion inherent to random matrix ensembles and the eigenvalue repulsion caused by $F(D)$ inspired the notation for $\beta_2$ as well as the normalization of $\beta_2$ in such a way that the values of $\beta_2$ correspond to those of the Dyson exponent for the complex sesquilinear and real variations of the fermionic action.

		Since the action is unitary invariant, the probability density can in fact be written entirely in terms of the eigenvalues of $H$ by Weyl integration, yielding
		\begin{align}
			p(\lambda) & = \frac{1}{Z_{N,\lambda}} \exp\left(\sum_{i,j} -2g_4 \left(m^2 + (\lambda_i - \lambda_j)^2\right)^2 - 2g_2 \left(m^2 + (\lambda_i - \lambda_j)^2\right) + \frac{\beta_2}{4} \log\left(m^2 + (\lambda_i - \lambda_j)^2\right) \right. \nonumber \\ & \hspace{15em} \left. + \sum_{i \neq j} \frac{\beta}{4} \log\left((\lambda_i - \lambda_j)^2\right) - a \left(\sum_i \lambda_i\right)^2 \right).
			\label{eq:fermionic_eigenvalue_model}
		\end{align}
		The appearance of a factor two in front of the $g_4$ and $g_2$ corresponds to the dimension of the fermion space.

		Note that for $m = 0$ the fermionic action adds to the Vandermonde repulsion, making these fermionic models examples of the so-called $\beta$-ensembles.
		For the $m = 0$ case with the reduced fermionic action we achieve a Dyson exponent of 3. In general we can get achieve integer Dyson exponent greater than 2 by adding a higher dimensional fermionic space with the reduced fermion action.
		General $\beta$-ensembles are the subject of active research \cite{bourgade2014edge,ramirez2011beta,forrester2010log}, providing a further application of these fermionic ensembles.

\section{The spectral density}
\label{sec:the_spectral_density}
	
	Sequences of random matrix ensembles of the form
	$$
		p_N(H) \propto \exp\left(-N\Tr(V(H))\right),
	$$
	with $V$ being a polynomial such that $\lim_{|x|\to\infty} V(x) = \infty$, are known \cite{johansson1998fluctuations} (see also \cite{deift1999orthogonal} for a more expository treatment) to have a probability measure $\mu_E$ such that for bounded functions $f$
	$$
		\lim_{N\to\infty} \langle \frac{1}{N}\Tr(f(H)) \rangle_{N} = \int_\R f(x) d \mu_E(x).
	$$
	This probability measure is the weak-$*$ limit of the expectation values of the spectral densities of $H$ as $N \to \infty$, and is therefore called the large $N$ spectral density or sometimes simply the spectral density of the matrix ensemble.
	In the following subsections, by generalizing the techniques from \cite{johansson1998fluctuations} we establish the existence of such a large $N$ spectral density for multi-tracial models with non-polynomial action. These generalizations, in particular, allow us to cover the fermionic quartic $(0, 1)$-fuzzy Dirac ensembles.

	\subsection{The equilibrium measure for multi-tracial matrix ensembles}
	\label{subsec:existence_of_eq_measure}
	
		Let us introduce some notation.
		\begin{notation}
			Let $k, N \in \N$, $k < N$, $f:\R^k \to \R$. Then define
			\begin{itemize}
				\item $[N] = \{1, \ldots, N\}$,
				\item $\Sym_N f:\R^N \to \R$ by $(\Sym_N f)(x) = \frac{1}{N^k}\sum_{i:[k] \to [N]} f(x_{i(1)}, \ldots, x_{i(k)})$,
				\item $\Per_N f:\R^N \to \R$ by $(\Per_N f)(x) = \frac{(N-k)!}{N!} \sum_{i:[k] \into [N]} f(x_{i(1)}, \ldots, x_{i(k)})$.
			\end{itemize}

			For a function $f:\R^k \to \R$ and $\mu$ a measure on $\R$ we define
			$$
				\langle f, \mu \rangle := \int_{\R^k} f(x) d\mu^{\otimes k}(x).
			$$

			Finally, all indices in otherwise unlabelled sums are assumed to run from $1$ to $N$ where $N$ will be clear from context.
		\end{notation}

		In this section we will consider an eigenvalue model that, at size $N$, has the form
		\begin{equation}
			p_N(\lambda) \propto \exp\left(-N^2 (\Sym_N U) (\lambda) + N(N-1)(\Per_N U_{Vdm})(\lambda)\right)
			\label{eq:abstract_eigenvalue_model}
		\end{equation}
		for 
		\[
			U_{Vdm}(x, y) = \frac{\beta}{4}\log\left((x-y)^2\right)
		\] 
		and some function $U:\R^k \to \R$. 
		The \emph{energy functional}, with values in $\R \cup \{\pm \infty\}$, corresponding to this eigenvalue model is defined by
		\begin{equation}
			I(\mu) = \langle (U - \Per_k U_{Vdm}), \mu \rangle = \int_{\R^k} U(x) - (\Per_k U_{Vdm})(x) d \mu^{\otimes k}(x)
			\label{eq:def_energy_functional}
		\end{equation}
		on any probability measure $\mu \in \mathcal{P}(\R)$.

		In order to obtain the existence of a spectral density in the large $N$ limit we impose the following assumptions on the function $U$.
		\begin{assumption}
			\label{ass:interaction_requirements}
			The matrix interaction $U:\R^k \to \R$ is a continuous function such that
			\begin{enumerate}
				\item $U$ is invariant under permutation of its arguments. Note that this can always be accomplished by replacing $U$ by $\Sym_k U$.
				\item There exists a function $u:\R \to \R$ such that
				\begin{enumerate}
					\item $U(x_1, \ldots, x_k) \geq u(x_1)$ for all $\vec{x} \in \R^k$,
					\item $u(x) - \frac{\max(\beta, 2)}{2} \log(1+x^2) \to \infty$ as $|x| \to \infty$,
				\end{enumerate}
				\item There is a set of ``candidate measures'' $\mathcal{P}_{can} \subset \mathcal{P}(\R)$ such that if
				$$
					I(\mu) = \inf_{\nu \in \mathcal{P}(\R)} I(\nu)
				$$ 
				then $\mu \in \mathcal{P}_{can}$, and for $\mu, \nu \in \mathcal{P}_{can}$ we have
				$$
					\frac{d^2}{dt^2} \int_{\R^k} U(\vec{x}) d\left((1-t)\mu + t\nu\right)^{\otimes k}(\vec{x}) \geq 0
				$$
				for $t \in [0,1]$.
			\end{enumerate}
		\end{assumption}

		\begin{remark}
			The formulation of the third assumption is the biggest deviation from \cite{johansson1998fluctuations}. The energy functional for the $(0, 1)$-fuzzy Dirac ensemble model is not necessarily convex on the space of probability measures, but it is convex between probability measures of the same mean. Using the inclusion of the non-physical probability density $\exp(-a\Tr(H)^2)$ on the fibers (see Section \ref{subsec:quartic_assoc_matrix_model}) we can show that any measure minimizing $I$ must have mean 0 so that we can set $\mathcal{P}_{can}$ to be probability measures of mean zero, on which the energy functional is convex.
		\end{remark}

		The first step is proving the existence of the equilibrium measure, a unique measure minimizing the associated energy functional $I$.

		\begin{theorem}
			\label{thm:equilibrium_measure_exists}
			Suppose a matrix interaction $U$ satisfies Assumption \ref{ass:interaction_requirements}. Then the energy functional $I$ of Equation \ref{eq:def_energy_functional} has a unique minimizing probability measure $\mu_E$ and this measure has compact support.
		\end{theorem}
		\begin{proof}[Sketch of proof.]
			The proof of this theorem is essentially identical to the proof of \cite[Thm. 2.1]{johansson1998fluctuations} so we include only the very broad strokes here. The proof first uses point 2 of Assumption \ref{ass:interaction_requirements} to establish that any sequence of measures approaching the infimum of $I$ is tight, \textit{i.e.} essentially confined to some compact subset of $\R$, so that it has a weak-$*$ limit point that minimizes $I$. Then the convexity property of Assumption \ref{ass:interaction_requirements} gives uniqueness. Full details of the required changes can be found in \cite[Section 3.2]{verhoeven2023thesis}.
		\end{proof}

		The next step is to show that this equilibrium measure $\mu_E$ is indeed the spectral density.
		Here we make a slightly larger deviation from \cite{johansson1998fluctuations,verhoeven2023thesis} to deal with both the multi-tracial nature of our models and to be able to deal with a larger class of observables beyond the bounded functions.
		
		\begin{definition}
			A sequence of functions $F_N:\R^N \to \R$ is called a \emph{tracial observable} if it is of the form
			$$
				F_N(x) = f\left( (\Sym_N f_1)(x), \ldots, (\Sym_N f_n)(x) \right)
			$$
			for some function $f:\R^n \to \R$ and $n$ functions $f_i:\R^{m_i} \to \R$.
		\end{definition}

		\begin{example*}
			Some relevant examples of tracial observables, in terms of $\lambda = (\lambda_1, \ldots, \lambda_N)$ the spectrum of an $N \times N$ matrix $H$, are
			\begin{itemize}
				\item Normalized products of traces of powers of a matrix, such as $\frac{1}{N^2}\Tr\left(H^3\right)\Tr(H)$:
				$$
					\frac{1}{N^2}\Tr\left(H^3\right)\Tr(H) = \frac{1}{N}\left(\sum_{i} \lambda_i^3\right) \frac{1}{N}\left(\sum_j \lambda_j\right) = \left(\Sym_N x^3\right)(\lambda) \cdot \left(\Sym_N x\right)(\lambda),
				$$
				or alternatively,
				$$
					\frac{1}{N^2}\Tr\left(H^3\right)\Tr(H) = \frac{1}{N^2} \sum_{i, j} \lambda_i^3 \lambda_j = \left(\Sym_N x^3y\right)(\lambda).
				$$
				\item Normalized traces of powers of $D = [H, \cdot]$, for example:
				$$
					\frac{1}{N^2}\Tr\left(D^2\right) = \frac{1}{N^2} \sum_{i,j} (\lambda_i - \lambda_j)^2 = \left(\Sym_N (x-y)^2\right)(\lambda).
				$$
				\item The spectral dimension, see Section \ref{sec:spectral_properties}, of $D = [H, \cdot]$ at $t \in \R$:
				$$
					d_s(t) := 2t \frac{\Tr\left(D^2 e^{-D^2t}\right)}{\Tr\left(e^{-D^2t}\right)} = 2t \frac{\sum_{i,j} (\lambda_i-\lambda_j)^2 e^{-(\lambda_i-\lambda_j)^2 t}}{\sum_{i,j} e^{-(\lambda_i-\lambda_j)^2 t}} = 2t \frac{\left(\Sym_N (x-y)^2 e^{-(x-y)^2t}\right)(\lambda)}{\left(\Sym_N e^{-(x-y)^2t}\right)(\lambda)}.
				$$
			\end{itemize}
		\end{example*}

		The power of Theorem \ref{thm:convergence_of_tracial_observables} is the ability to move the expectation values into $f$. This is a significant generalization of the factorization of moments
		$$
			\lim_{N \to \infty} \langle \frac{1}{N^2}\Tr(H^k)\Tr(H^l) \rangle = \left(\lim_{N \to \infty} \langle \frac{1}{N}\Tr(H^k) \rangle\right)\left(\lim_{N \to \infty} \langle \frac{1}{N}\Tr(H^l) \rangle\right)
		$$
		observed in \cite{hessam2022bootstrapping}.

		\begin{theorem}
			\label{thm:convergence_of_tracial_observables}
			Given a sequence of eigenvalue models as in Equation \ref{eq:abstract_eigenvalue_model} satisfying Assumption \ref{ass:interaction_requirements} and a tracial observable $F_N:\R^N \to \R$ such that
			\begin{enumerate}
				\item there is a function $g:\R^n \to \R$ such that $|F_N(x)| \leq (\Sym_N g)(x)$ for all $N$ and $x \in \R^N$ and the integral
				$$
					\int_{\R^n} g(x) e^{-\sum_{i=1}^n u(x_i)} d^n x
				$$
				is bounded,
				\item the constituent functions $f_1, \ldots, f_n$ are bounded,
				\item for any probability measure $\nu$ the function $f$ is continuous at the point $\left(\langle f_1, \nu\rangle, \ldots, \langle f_n, \nu\rangle \right)$.
			\end{enumerate}
			We then have
			$$
				\lim_{N \to \infty} \langle F_N \rangle_{p_N} = f\left(\langle f_1, \mu^E \rangle, \ldots, \langle f_n, \mu^E \rangle \right).
			$$
		\end{theorem}

		Before turning our attention to the technical details of this proof, let us establish that this theorem covers at least the cases we are most interested in.

		\begin{proposition}
			\label{prop:our_model_satisfies_assumptions}
			The function
			$$
				U(x, y) = 2g_4 \left(m^2 + (x-y)^2\right)^2 + 2g_2 \left(m^2 + (x-y)^2\right) - \frac{\beta_2}{4} \log\left(m^2 + (x-y)^2\right) + \frac{a}{2}xy
			$$
			defining the fermionic quartic $(0, 1)$-fuzzy geometry satisfies Assumption \ref{ass:interaction_requirements}.
		\end{proposition}
		\begin{proof}
			Note that $U(x+t, y+t) = U(x, y) - axy + a(x+t)(y+t) = U(x, y) + a(x+y)t + at^2$.
			Hence
			\begin{align*}
				\inf_{y \in \R} U(x, y) & = \inf_{y \in \R} U(0, y-x) + axy, \\
				& = \inf_{y \in \R} U(0, y) + ax(y+x), \\
				& \geq ax^2 + \inf_{y \in \R}  \left(2g_4 m^4 + 4 g_4 m^2 y^2 + g_4 y^4 + 2g_2(m^2 + y^2) - \frac{\beta_2}{4} \log\left(m^2 + y^2\right)\middle) + \inf_{y \in \R} \middle(g_4 y^4 + axy\right).
			\end{align*}
			By some straightforward calculus, the third term is proportional to $x^{\frac{4}{3}}$.
			The second term is independent of $x$ and finite so that as a whole
			$$
				U(x, y) \geq \frac{1}{2}a x^2,
			$$
			at least for $x$ large enough. 
			This proves the second point of Assumption \ref{ass:interaction_requirements}.

			For the third point observe that for any probability measure $\mu$ we have
			\begin{align*}
				I(\mu) & = \langle U, \mu \rangle = \int_{\R^2} U(x, y) d\mu(x)d\mu(y), \\
				& = \int_{\R^2} U(x + t, y + t) - a(x+y)t + at^2 d\mu(x)d\mu(y), \\
				& = \int_{\R^2} U(x, y) d\mu(x-t) d\mu(y-t) - 2at \langle x, \mu \rangle + at^2.
			\end{align*}
			This means that among a family of translates of a given probability measure, the translate with mean equal to zero will minimize the energy functional.

			Splitting $U$ into the polynomial part and logarithmic part
			\begin{align*}
				U_{pol}(x, y) & = 2g_4 \left(m^2 + (x-y)^2\right)^2 + 2g_2 \left(m^2 + (x-y)^2\right) + \frac{a}{2}xy, \\
				U_{log}(x, y) & =  - \frac{\beta_2}{4} \log\left(m^2 + (x-y)^2\right),
			\end{align*}
			both pairings can be computed separately.

			The first pairing
			$$
				\langle U_{log}, \nu - \mu \rangle = \frac{\beta_2}{2} \int_0^\infty \frac{1}{k}e^{-mk^2} \left|\widehat{(\nu-\mu)}(k)\right|^2 dk \geq 0,
			$$
			where $\widehat{\nu-\mu}$ denotes the Fourier transform of $\nu-\mu$ \cite[Prop. 3.1.6]{verhoeven2023thesis}.
			The polynomial pairing on the other hand can be computed to be 
			$$
				\langle U_{pol}, \nu - \mu \rangle = 12g_4 \langle x^2, \nu - \mu \rangle^2 - 16 g_4 \langle x^3, \nu - \mu \rangle \langle x, \nu - \mu \rangle + \left(-4g_2 - 8m^2g_4 + \frac{a}{2}\right) \langle x, \nu-\mu \rangle^2.
			$$
			If $\nu$ and $\mu$ are both potential minimizers their means agree, $\langle x, \nu - \mu \rangle = 0$.
			Hence also $\langle U_{pol} , \nu - \mu \rangle \geq 0$ and all three assumptions are satisfied.
		\end{proof}

		\begin{proposition}
			The spectral dimension and spectral variance (see Section \ref{sec:spectral_properties}) of $D = [H, \cdot]$ can be computed for the fermionic quartic $(0, 1)$-fuzzy geometry by Theorem \ref{thm:convergence_of_tracial_observables} from the equilibrium measure for $H$.
		\end{proposition}
		\begin{proof}
			The spectral dimension is the tracial observable
			$$
				F_N(\lambda) = d_s(t)(\lambda) = 2t\frac{\sum (m^2 + (\lambda_i - \lambda_j)^2) e^{-t (m^2 + (\lambda_i - \lambda_j)^2)}}{\sum e^{-t (m^2 + (\lambda_i - \lambda_j)^2)}},
			$$
			so it has $f(x, y) = \frac{x}{y}$ and constituent functions $f_1(x, y) = (m^2 + (x-y)^2) e^{-t(m^2 + (x-y)^2)}$, $f_2(x, y) = e^{-t(m^2 + (x-y)^2)}$ which are certainly bounded for $t > 0$.
			In order to apply Theorem \ref{thm:convergence_of_tracial_observables} we need to establish that $f$ is continuous at all possible points $(\langle f_1, \nu \rangle, \langle f_2, \nu \rangle)$ with $\nu$ as a probability measure, and that $F_N$ has sufficient decay compared to $e^{-u(x)}$.

			Starting with this last point, it is straightforward to check that 
			$$
				F_N(\lambda) \leq 2t \sum (m^2 + (\lambda_i - \lambda_j)^2) = (\Sym_N 2t(m^2 + (x-y)^2))(\lambda),
			$$
			so that we can take $g(x, y) = 2t(m^2 + (x-y)^2)$.
			From the proof of Proposition \ref{prop:our_model_satisfies_assumptions} we know that $u(x) = \frac{1}{2}ax^2$, and certainly
			$$
				\int_{\R^2} 2t(m^2 + (x-y)^2) e^{-\frac{1}{2}a(x^2 + y^2)} dx dy < \infty.
			$$

			Next we prove the continuity of $f$. 
			For this we need to establish that $\langle f_2, \nu \rangle \neq 0$ if $\nu$ is a probability measure.
			If $\nu$ is a probability measure there is some $R$ such that $\int_{-R}^R d\nu \geq \frac{1}{2}$.
			Then we have
			\begin{align*}
				\langle f_2, \nu \rangle & = \int_{\R^2} e^{-t(m^2 + (x-y)^2)} d\nu(x)d\nu(y), \\
				& \geq \int_{[-R, R]^2} e^{-t(m^2 + (x-y)^2)} d\nu(x)d\nu(y), \\
				& \geq \int_{[-R, R]^2} e^{-t(m^2 + 4R^2)} d\nu(x)d\nu(y), \\
				& \geq \frac{1}{4} e^{-t(m^2 + 4R^2)} > 0.
			\end{align*}

			The same strategy works for the spectral variance.
		\end{proof}

	\subsection{Proof of Theorem \ref{thm:convergence_of_tracial_observables}}
	\label{subsec:proof_that_eq_measure_is_spec_density}

		The main gadget involved in this proof is the sequence of sets
		\begin{equation}
			A_{N, \eta} := \left\{ x \in \R^N \, \middle| \, (\Sym_N U)(x) - (\Per_N U_{Vdm})(x) \leq E + \eta \right\}
			\label{eq:def_of_A}
		\end{equation}
		with $E = \inf I(\mu)$.

		\begin{lemma}
			The sets $A_{N, \eta}$ are non-empty and compact.
		\end{lemma}
		\begin{proof}
			By Theorem \ref{thm:equilibrium_measure_exists} there exists a measure $\mu_E$ such that
			$$
				E = I(\mu_E) = \int_{\R^k} U(x) - (\Per_k U_{Vdm})(x) d \mu_E^{\otimes k}(x) = \int_{\R^N} (\Sym_N U)(x) - (\Per_N U_{Vdm})(x) d \mu_E^{\otimes N}(x)
			$$
			hence there is a point in $(\supp \mu_E)^N$ such that $(\Sym_N U)(x) - (\Per_N U_{Vdm})(x) \leq E$.

			These sets are bounded, since
			\begin{align*}
				(\Sym_N U)(x) - (\Per_N U_{Vdm})(x) & = \frac{1}{N^k} \sum_{i:[k] \to [N]} U(x_{i(1)}, \ldots, x_{i(k)}) - \frac{1}{N(N-1)} \sum_{i:[2] \into [N]} \frac{\beta}{2} \log(|x_{i(1)} - x_{i(2)}|), \\
				& \geq \frac{1}{N^k} \sum_{i:[k] \to [N]} u(x_{i(1)}) - \frac{1}{N(N-1)} \sum_{i:[2] \into [N]} \frac{\beta}{2} \left(\log\left(\sqrt{1+x_{i(1)}^2}\right) + \log\left(\sqrt{1+x_{i(2)}^2}\right)\right), \\
				& \geq \frac{1}{N} \sum_{i=1}^N u(x_i) - \frac{1}{N} \sum_{i=1}^N \frac{\beta}{2} \log\left(1+x_i^2\right), \\
				& = \frac{1}{N} \sum_{i=1}^N \left(u(x_i) - \frac{\beta}{2} \log(1+x_i^2)\right)
			\end{align*}
			and the summands tend to go to infinity by point 2 of assumption \ref{ass:interaction_requirements}. The sets $A_{N, \eta}$ are also closed, hence they are compact.
		\end{proof}

		The first step is to show that expectation values of $e^{u(x)}$-bounded observables are determined by their behaviour on $A_{N, \eta}$.
		This is a large deviation type result which was proven for non-interacting polynomial potentials in \cite{johansson1998fluctuations}.

		\begin{lemma}
			\label{lem:expectation_values_restrict}
			Suppose that $\{F_N\}_N$ is a sequence of observables that is bounded by a function $g:\R^m \to \R$, in the sense that for all $N$ and $x \in \R^N$ we have $|F_N(x)| \leq (\Sym_N g)(x)$, and $g$ satisfies
			$$
				\int_{\R^m} g(x) e^{-\sum_{i=1}^m u(x_i)} d^m x < \infty
			$$
			and is symmetric under permutation of its arguments.
			Then,
			$$
				\int_{\R^N \setminus A_{N, \eta}} F_N(x) p_N(x) d^N x \to 0.
			$$
		\end{lemma}

		\begin{proof}
			Let $\eta > 0$ be arbitrary and $d\mu(x) = \phi(x)dx$ a measure with $I(\mu) < E + \eta$. This exists by smoothing the minimizer and a continuity result for $I$ (\cite[Lemma 3.2.12]{verhoeven2023thesis}).

			Following the proof of \cite[Prop. 3.2.13]{verhoeven2023thesis} (Lemma 4.2 in \cite{johansson1998fluctuations}), we can show that
			$$
				\frac{1}{Z_N} \leq e^{N^2(E + \frac{\eta}{2})}.
			$$

			This implies that
			\begin{align*}
				\left|\int_{\R^N \setminus A_{N, \eta}} (\Sym_N g)(x) p_N(x) d^N x\right| & = \left|\frac{1}{Z_N} \int_{\R^N \setminus A_{N, \eta}} F_N(x) e^{-S_N(x)} d^N x \right|, \\
				& \leq e^{N^2(E + \frac{\eta}{2})} \left|\int_{\R^N\setminus A_{N, \eta}} (\Sym_N g)(x) e^{-N^2 (\Sym_N U)(x) + N(N-1)(\Per_N U_{Vdm})(x)} d^N x\right|, \\
				& \leq e^{N^2(E + \frac{\eta}{2})} \left|\int_{\R^N\setminus A_{N, \eta}} (\Sym_N g)(x) e^{-N (\Sym_N U)(x) - N(N-1)(E+\eta))(x)} d^N x\right|, \\
				& = e^{-\frac{\eta}{2}N^2 + (E + \eta)N} \left|\int_{\R^N \setminus A_{N, \eta}} (\Sym_N g)(x) e^{-N (\Sym_N U)(x)} d^N x\right|.
			\end{align*}
			
			The remaining integral can be bounded by
			\begin{align*}
				\int_{\R^N \setminus A_{N, \eta}} (\Sym_N g)(x) e^{-N (\Sym_N U)(x)} d^N x & \leq \int_{\R^N} (\Sym_N g)(x) e^{-N(\Sym_N u)(x)} d^N y, \\
				& = \int_{\R^m} g(y) e^{-\sum_{i=1}^m u(y_i)} d^m y \int_{\R^{N-m}} e^{-\sum_{i=1}^{N-m} u(z_i)} d^{N-m}z, \\
				& = \frac{\int g(y) e^{-\sum u(y_i)} d^m y}{\left(\int e^{-u(z)} dz\right)^m} \left(\int_\R e^{-u(z)} dz\right)^N,
			\end{align*}
			which by the assumption on $g$ and the growth condition on $u$ can be written as $C e^{c N}$ for some finite real numbers $c, C$.

			Thus 
			$$
				\int_{\R^N \setminus A_{N, \eta}} (\Sym_N g)(x) p_N(x) d^N x \leq C e^{-\frac{\eta}{2}N^2 + (E + \eta + c)N} \to 0,
			$$
			and this implies the statement for $F_N$.
		\end{proof}

		\begin{lemma}
			Let $\{y_N\}_{N = k}^\infty$ be a sequence of points such that $y_N \in A_{N, \eta}$.
			Then the sequence of measures 
			$$
				\left\{\nu_N = \frac{1}{N} \sum_{i=1}^N \delta_{y_i}\right\}
			$$
			has a weak limit point $\nu_\eta$ with $I(\nu_\eta) \leq E + \eta$.
			\label{lem:measures_for_AN_converge}
		\end{lemma}
		\begin{proof}
			The sequence of measures $\{\nu_N\}$ is tight. Indeed, choose any $\epsilon > 0$. Let $b$ be a lower bound for $u(x) - \frac{\beta}{2}\log(1 + x^2)$ and $B$ such that $(1 - \epsilon)b + \epsilon B > E + \eta$.
			There is some $R$ such that if $|x| > R$ then $u(x) - \frac{\beta}{2}\log(1+x^2) > B$. Then at most $\epsilon N$ of the coordinates in any $x \in A_{N, \eta}$ lie outside of $[-R, R]^N$ so that at most a mass $\epsilon$ lies outside of $[-R, R]$ for $\nu_N$.
			This implies that the sequence $\{\nu_N\}$ has at least one weak limit point.

			Any such limit point $\nu$ has minimal energy. Let $l(x) = \frac{\beta}{2}\log\left(1+x^2\right)$ and note that for functions of a single argument $\Per_k$ and $\Sym_k$ are the same.
			\begin{align*}
				I(\nu) & = \lim_{L \to \infty} \lim_{N \to \infty} \int_{\R^k} \min(L, U(x) - (\Per_k U_{Vdm})(x)) d \nu_N^{\otimes k}(x), \\
				& = \lim_{L \to \infty} \lim_{N \to \infty} N^{-k} \sum_{i:[k] \to [N]} \min\left(L, U - \Per_k U_{Vdm}\right)(y_{N,i}), \\
				& \leq \lim_{L \to \infty} \lim_{N \to \infty} N^{-k} \sum_{i:[k] \to [N]} U(y_{N,i}) - \Per_k l + \min\left(L-b, \Per_k l - \Per_k U_{Vdm}\right)(y_{N,i}), \\
				& \leq \lim_{L \to \infty} \lim_{N \to \infty} N^{-k} \sum_{i:[k] \to [N]} (U(y_{N,i}) - \Sym_k l) + N^{-k} \sum_{i:[k]\into[N]} (\Per_k l - \Per_k U_{Vdm})(y_{N,i}) + N^{-k} \sum_{\substack{i:[k] \to [N] \\ i\text{ not injective}}} L-b, \\
				& \leq \lim_{L \to \infty} \lim_{N \to \infty} N^{-k} \sum_{i:[k] \to [N]} (U(y_{N,i}) - \Sym_k l) + \frac{(N-k)!}{N!} \sum_{i:[k]\into[N]} (\Per_k l - \Per_k U_{Vdm})(y_{N,i}), \\
				& \leq \lim_{L \to \infty} \lim_{N \to \infty} E + \eta.
			\end{align*}
			Here the function $l$ is ``moved'' between $U$ and $U_{Vdm}$ to ensure that the term involving $-U_{Vdm}$ is positive while the term involving $U$ is still bounded below. This allows us to move the $\min$ to the term involving $-U_{Vdm}$ and ensure that $N^{-k}$ may be replaced by the larger coefficient $(N-k)!/N!$. Together these steps replace the full symmetrization $\Sym_N \Per_k U_{Vdm}$ coming from the use of $\nu_N^{\otimes k}$ with the injective symmetrization $\Per_N \Per_k U_{Vdm} = \Per_N U_{Vdm}$, which allows the use of the definition of $A_{N, \eta}$.
		\end{proof}

		We are now ready to prove our main theorem on convergence of spectral observables.

		\begin{proof}[Proof of Theorem \ref{thm:convergence_of_tracial_observables}]

			Let $F_N = f\left(\Sym_N f_1, \ldots, \Sym_N f_n\right)$ be the tracial observable under consideration. Fix an $\eta > 0$.
			By Lemma \ref{lem:expectation_values_restrict} we have
			$$
				\limsup_{N \to \infty} \langle F_N \rangle_{p_N} = \limsup_{N \to \infty} \langle F_N \cdot \chi_{A_{N, \eta}} \rangle_{p_N}.
			$$
			Let $y_{N, \eta}$ be a point where $F_N$ achieves its maximum on $A_{N, \eta}$ and $\nu_{N, \eta}(x) = \frac{1}{N} \sum_{i} \delta(x - y_{N, \eta, i})$.
			We arrive at
			$$
				\limsup_{N \to \infty} \langle F_N \rangle_{p_N} \leq \limsup_{N \to \infty} F_N(y_{N, \eta}) = \limsup_{N \to \infty} f\left(\int_{\R^{m_1}} f_1(x) d\nu_{N, \eta}^{\otimes m_i}(x), \ldots \right).
			$$

			By Lemma \ref{lem:measures_for_AN_converge} applied to a subsequence realizing this $\limsup$ we get a weak limit point $\nu_\eta$ of $\{\nu_N\}$. This $\nu_\eta$ satisfies
			$$
				\limsup_{N \to \infty} \langle F_N \rangle_{p_N} \leq f\left(\langle f_1, \nu_{\eta} \rangle, \ldots, \langle f_n, \nu_{\eta} \rangle \right),
			$$
			using the continuity property of $f$ and the boundedness of the constituent functions $f_i$.
			Note that by said Lemma we also have $I(\nu_\eta) \leq E + \eta$.

			Letting $\eta \to 0$ we can construct a tight sequence of measures $\{\nu_\eta\}$ which therefore has a weak limit point $\nu$. 
			By weak lower semicontinuity of $I$ \cite[Lem. 3.2.8]{verhoeven2023thesis} we obtain $I(\nu) \leq E$ so that the uniqueness of the minimizer of $I$ implies that $\nu_\eta \weakto \mu_E$.
			Hence
			$$
				\limsup_{N \to \infty} \langle F_N \rangle_{p_N} \leq f\left(\langle f_1, \mu_E \rangle, \ldots, \langle f_n, \mu_E \rangle \right).
			$$
			Repeating this with $\liminf$ instead of $\limsup$ (and reversal of appropriate inequalities) we get
			$$
				\lim_{N \to \infty} \langle F_N \rangle_{p_N} = f\left(\langle f_1, \mu_E \rangle, \ldots, \langle f_n, \mu_E \rangle \right)
			$$
			so we obtain the desired convergence.
		\end{proof}

	\subsection{Computation of the equilibrium measure}
	\label{subsec:computation_of_eq_measure}

		By Theorems \ref{thm:equilibrium_measure_exists} and \ref{thm:convergence_of_tracial_observables} we can compute certain observables of our Dirac ensembles in the large $N$ limit using the equilibrium measure. The equilibrium measure minimizes the energy functional
		\begin{equation}
			I(\mu) = \int \left ( 2g_4 (x-y)^4 + 2g_2  (x-y)^2 - \frac{\beta_2}{4} \log\left(m^2 + (x-y)^2\right) - \frac{\beta}{4} \log\left((x-y)^2\right) + \frac{a}{2}xy\right ) d\mu(x) d\mu(y)
			\label{eq:energy_functional_expanded}
		\end{equation}
		and can be found using essentially the same techniques as in the single-trace case as found in for example \cite[Ch. 6.7]{deift1999orthogonal}. See also \cite{verhoeven2023thesis}.
		
		In order to find $\mu_E$ we assume that $\mu_E = \rho(x)dx$ for $\rho:\Sigma \to \R$ and $\Sigma \subset \R$ some compact subset.
		Since we know the minimizer is unique and the problem is symmetric under $x \mapsto -x$ we can further restrict our attention to symmetric functions $\rho$.
		The main step is then to transform the minimization problem for the functional in Equation \ref{eq:energy_functional_expanded} into the variational problem
		\begin{align*}
			\PV \int_\Sigma \frac{\rho(y)}{y-x} dy & = -\frac{2}{\beta} \frac{d}{dx} \int_\Sigma \left( 2g_4 (x-y)^4 + 2g_2  (x-y)^2 - \frac{\beta_2}{4} \log\left(m^2 + (x-y)^2\right) + \frac{a}{2}xy \right) \rho(y)dy, \\
			& = -\frac{2}{\beta} \int_\Sigma \left(8 g_4 (x-y)^3 + 4 g_2 (x-y) - \frac{\beta_2}{2} \frac{x-y}{m^2 + (x-y)^2} + \frac{a}{2}y \right) \rho(y) dy, \\
			& = -\frac{2}{\beta} \left(8g_4 x^3 + (24 g_4 \mu^E_2 + 4 g_2) x \right)- \frac{\beta_2}{2} \int_\Sigma \frac{x-y}{m^2 + (x-y)^2} \rho(y)dy .
		\end{align*}
		
		This equation for $\rho$ can be transformed further by using the Stieltjes transform and  the Sokhotski-Plemelj formula for Riemann-Hilbert problems. 
		Details can again be found in \cite{deift1999orthogonal,verhoeven2023thesis}.
		The result of these transformations are the following equations depending on the form of the support $\Sigma$:
		\begin{description}
			\item[{$\Sigma = [-a, a]$: }]
			\begin{subequations}
				\label{eq:one_cut_problem}
				\begin{align}
					& \rho(x) = \frac{\sqrt{a^2-x^2}}{\pi} \left(\frac{2}{\beta}  \left(8g_4 x^2 + \left(4g_4 a^2 + 24g_4 \mu^E_2 + 4g'_2\right)\right) + \frac{\beta_2}{\beta} \int_{\Sigma} K_{\Sigma, m}(x, y) \rho(y) d y\right), \label{subeq:one_cut_problem_rho} \\
					& \beta + \beta_2\left(1-\int_{\Sigma} R_{\Sigma, m}^1(y) \rho(y) d y\right) = 6g_4 a^4 + \left(24g_4\mu^E_2 + 4g'_2\right)a^2, \label{subeq:one_cut_problem_normalization} \\
					& \mu^E_2 = \int_{\Sigma} x^2 \rho(x) d x, \label{subeq:one_cut_problem_consistency}
				\end{align}
			\end{subequations}
			
			\item[{$\Sigma = [-a, -b] \cup [b, a]$: }]
			\begin{subequations}
				\label{eq:two_cut_problem}
				\begin{align}
					& \rho(x) = \frac{\sqrt{(x^2-b^2)(a^2-x^2)}}{\pi} \left(\frac{2}{\beta}  8 g_4 |x| + \frac{\beta_2}{\beta} \int_{\Sigma} K_{\Sigma, m}(x, y) \rho(y) d y\right), \label{subeq:two_cut_problem_rho} \\
					& 0 = \beta_2 \int_{\Sigma} R_{\Sigma, m}^0(y) \rho(y) d y + 8g_4\left(a^2+b^2\right) + \left(48g_4 \mu^E_2 + 8g_2'\right), \label{subeq:two_cut_problem_moment_zero} \\
					& \beta + \beta_2\left(1 - \int_{\Sigma} R_{\Sigma, m}^2(y) \rho(y) d y\right) = 2g_4\left(3a^4 + 2a^2b^2 + 3b^4\right) + \left(24g_4 \mu^E_2 + 4g_2'\right)\left(a^2+b^2\right), \label{subeq:two_cut_problem_normalization} \\
					& \mu^E_2 = \int_{\Sigma} x^2 \rho(x) d x, \label{subeq:two_cut_problem_consistency}
				\end{align}
			\end{subequations}
		\end{description}
		where
		\begin{align*}
			K_{\Sigma, m}(x, y) & = \Re\left(\frac{1}{\sqrt{s}(y+im)} \frac{1}{y+im-x}\right), \\
			R_{\Sigma, m}^n(y) & = \Re\left(\frac{(y+im)^n}{\sqrt{s}(y+im)}\right), \\
			\sqrt{s}(z) & = \left\{\begin{array}{l}
				\sqrt{z^2-a^2}, \\
				\sqrt{(z^2-a^2)(z^2-b^2)}.
			\end{array}\right.
		\end{align*}
		The square root is chosen such that $\sqrt{s}(x)$ is positive and real for $x > 0$ and that the function has branch cuts on $\Sigma$.

		These systems consist of a Fredholm integral equation for $\rho$, Equations \ref{subeq:one_cut_problem_rho}, \ref{subeq:two_cut_problem_rho}, and equations that can be used to determine the parameters of $\Sigma$, \ref{subeq:one_cut_problem_normalization}, \ref{subeq:two_cut_problem_moment_zero}, \ref{subeq:two_cut_problem_normalization}.
		Note that if $m = 0$ these equations simplify, since the integral term involving $K_{\Sigma, m}$ can be folded into the original principal value.
		This has the ultimate effect of removing the $K_{\Sigma, m}$ and $R_{\Sigma, m}$ terms from the equations and replacing $\beta$ by $\beta + \beta_2$; thus if $m = 0$, we obtain straightforward equations for $\rho$ and the parameters $a, b$ determining the support of $\rho$.
		
		It turns out that these cases for $\Sigma$ are exhaustive for quartic polynomials in $D$.
		The case $\Sigma = [-a, a]$ is referred to as the one-cut phase for the Dirac ensemble, and the case $\Sigma = [-a, -b] \cup [b, a]$ as the two-cut phase.
		These phases and their transitions are further explored in Section \ref{sec:spectral_properties}.
		When considering hexic or even higher degree potentials more of these ``spectral phases'' need to be considered. In general a polynomial of degree $2n$ can correspond to a support $\Sigma$ consisting of up to $n$ intervals, see also \cite{hessam2022noncommutative}.
		We keep our investigations to the quartic case for this paper since it already exhibits two spectral phases.

		The equilibrium measures for $m = 0$ are explored in Section \ref{subsec:estimators_massless}, after discussing the properties of Dirac ensembles that can be derived from these equilibrium measures.
		In Section \ref{subsec:estimators_non_zero_mass} we will numerically investigate these equilibrium measures for $m \neq 0$ and the effect on the properties of the Dirac ensembles caused by the presence of a mass.

\section{Spectral properties of the Dirac ensemble}
\label{sec:spectral_properties}

	One of the main tools to access geometric data encoded in the spectrum of a Dirac operator $D$ is the heat kernel 
	$$
		K_{D^2}(t) = \sum_{\lambda \in \spec(D)} e^{-t\lambda^2}.
	$$
	For the Dirac operator on a manifold $M$ of dimension $d$ this converges for $t > 0$ as the eigenvalues tend to approach infinity at a rate of $\lambda_n \sim n^{1/d}$ \cite{berline2003heat}.

	Moreover, as $t \to 0$ the heat kernel has an asymptotic expansion
	$$
		K_{D^2}(t) \sim_{t \to 0} t^{-d/2}(a_0 + a_2 t + a_4 t^2 + \ldots),
	$$
	where the $a_i$ are spectral invaraints. For example, $a_0 = (4\pi)^{-d/2} \vol(M)$ and higher $a_i$ involve the Riemann curvature and its derivatives \cite{gilkey2018invariance}.
	Using this heat kernel expansion is a way to define  geometric invariants such as dimension, volume, and total scalar curvature \cite{connes2000short,van2015noncommutative,Fathizadeh2019} for  spectral triples.

	However, for our models the spectral density remains bounded so that in the large $N$ limit the heat kernel is undefined for any $t$.
	Instead we can compute the large $N$ limit of the \emph{normalized heat kernel}
	$$
		k_{D^2}(t) := \lim_{N \to \infty} \langle \frac{1}{N^2} \sum_{\lambda \in \spec(D_N)} e^{-t\lambda^2} \rangle_{p_N}.
	$$
	This normalized heat kernel clearly lacks any divergence at $t = 0$ but it can still be used to attach geometric information to our ensembles of bounded Dirac operators by using different methods to extract geometric data from the spectrum.

	The two notions we will use in this paper are the spectral dimension and spectral variance (we collectively call these spectral estimators), as used in \cite{barrett2019spectral} for Monte-Carlo simulations of similar fuzzy geometries.
	For a Dirac operator $D$ with spectrum $\{\lambda_i\}_{i = 1}^N$, the spectral dimension at energy $t$ is defined by
	\begin{equation}
		d_s(t) := -2t \frac{d \log(k_{D^2}(t))}{dt} = 2t \frac{\sum \lambda_i^2 e^{-t \lambda_i^2}}{\sum e^{-t\lambda_i^2}}.
		\label{eq:spectral_dimension}
	\end{equation}
	By using the logarithmic derivative we can essentially sidestep the problem that we only have access to the normalized heat kernel, as the normalization is a constant shift of $\log k_{D^2}(t)$ versus $\log K_{D^2}(t)$. For manifolds $\lim_{t \to 0} d_s(t)$ gives the dimension, while the $t \to \infty$ limit is 0 if the space is compact and depends on the dimension if it is not \cite{barrett2019spectral}.

	Observe that in the case of a fermionic $(0, 1)$-fuzzy geometry with mass $m$ the spectrum of the full Dirac operator consists of $\{\pm \sqrt{m^2 + \lambda_i^2}\}$ for $\{\lambda_i\}$, the spectrum of the \emph{fuzzy} Dirac operator.
	In this case the spectral dimension can be computed to be
	\begin{align*}
		d_s(t) & = 2t \frac{2\sum (m^2 + \lambda_i^2) e^{-t (m^2 + \lambda_i^2)}}{2\sum e^{-t(m^2 + \lambda_i^2)}}, \\
		& = 2t \left(m^2 + \frac{\sum \lambda_i^2 e^{-t \lambda_i^2}}{\sum e^{-t\lambda_i^2}}\right), \\
		& \sim_{t \to \infty} 2tm^2
	\end{align*}
	because 0 is always an eigenvalue of $D = [H, \cdot]$.

	As noted in \cite{barrett2019spectral} it can be useful to further refine the spectral dimension to the spectral variance
	\begin{equation}
		v_s(t) := 2t^2 \frac{d^2 \log(k_{D^2}(t))}{dt^2} = 2t^2 \left(\frac{\sum \lambda_i^4 e^{-t \lambda_i^2}}{\sum e^{-t\lambda_i^2}} - \left(\frac{\sum \lambda_i^2 e^{-t \lambda_i^2}}{\sum e^{-t\lambda_i^2}}\right)^2\right).
	\end{equation}
	The spectral variance removes the linear behaviour for $t \to \infty$ from the spectral dimension caused by having a non-zero lowest eigenvalue. It was originally introduced to study dynamical triangulations \cite{ambjorn2005spectral} and has since been used in other models of quantum gravity \cite{carlip2017dimension}. For closed manifolds, just as with spectral dimension, the spectral variance limit as $t$ goes to infinity is zero. The general behaviour for manifolds of the spectral variance as $t$ goes to zero is harder to analyze.

	In the next sections we will compute these dimension estimators for the quartic type $(0, 1)$-fuzzy Dirac ensemble and explore how they are affected by various values of our coupling constants. When computing the spectral estimators $d_s(t)$ and $v_s(t)$ for an ensemble of Dirac operators it is important to clarify where the expectation value is taken.
	Corresponding to the discussion in \cite[Sec. D1]{barrett2019spectral} one can either compute the ``average spectrum'' and find the estimators for this average spectrum, or one can look for the expected value of the spectral estimators over the ensemble.
	In the setting of Monte-Carlo simulations, at a certain matrix size $N$ the average spectrum can be computed by recording the $n$-th eigenvalue (ordered from $-\infty$ to $\infty$, say) for every observed matrix in the Monte-Carlo run. The $n$-th average eigenvalue is then simply the average of these observed $n$-th eigenvalues.

	Theorem \ref{thm:convergence_of_tracial_observables} allows us to use the equilibrium measure to compute the expected values of the estimators in the large $N$ limit, since it tells us that (for $D = [H, \cdot]$, $\spec H = \{\lambda_i\}$)
	$$
		\lim_{N \to \infty} \langle d_s(t) \rangle = \lim_{N \to \infty} \langle 2t \frac{N^{-2} \sum (\lambda_i - \lambda_j)^2 e^{-t(\lambda_i - \lambda_j)^2}}{N^{-2} \sum e^{-t(\lambda_i - \lambda_j)^2}}\rangle = 2t \frac{\int_{\R^2} (x-y)^2 e^{-t(x-y)^2} d\mu_E(x) d\mu_E(y)}{\int_{\R^2} e^{-t(x-y)^2} d\mu_E(x) d\mu_E(y)}.
	$$
On the other hand the equilibrium measure is constructed as the weak limit of the expectation value of spectral densities, and as such is the large $N$ limit of the ``average spectrum'' if the average spectrum is taken to be the expected value of the spectral density.
	It is unclear at this time how this relates to the average spectral density for Monte-Carlo simulations in \cite{barrett2019spectral}.


	For flat space the spectral dimension and variance converge to the dimension of the space in the low-energy limit $t \to \infty$, as they do in the presence of a high-energy cutoff; see equations (14) and (19) of \cite{barrett2019spectral}.
	On compact spaces the low-energy limit is zero, as the corresponding long wavelengths no longer ``fit'' on the space.
	Generally the low-energy limit is determined by the smallest eigenvalues. In the compact case there is a smallest eigenvalue while in the non-compact case there is a non-discrete spectrum with no single smallest eigenvalue and the limit is determined by the density near 0.

	For a given spectral density the low-energy limit can be understood using the following elementary Lemma akin to the initial value theorem for the Laplace transform.
	\begin{lemma}
		\label{lem:initial_value_theorem}
		Suppose $f$, $g$ are functions such that $\int_\R f(x) dx$ is absolutely convergent, $g$ is bounded, and $\lim_{x \to 0} g(x) = L$ is finite.
		Then $\lim_{\alpha \to 0} \int_\R f(x) g(\alpha x) dx = L \int_\R f(x) dx$.
	\end{lemma}

	Suppose $\lim_{x \to 0} x^{-a} \rho_{D}(x) = L$, $0 < L < \infty$, then we can use the above Lemma to compute
	\begin{align} 
		\lim_{t \to \infty} d_s(t) & = 2 \lim_{t \to \infty} \frac{\int_\R tx^2 e^{-tx^2} \rho_D(x) dx}{\int_\R e^{-tx^2} \rho_D(x) dx}, \nonumber \\
		& = 2 \lim_{t \to \infty} \frac{\int_\R s^{2+a} e^{-s^2} \left(\frac{s}{\sqrt{t}}\right)^{-a} \rho_D\left(\frac{s}{\sqrt{t}}\right) ds}{\int_\R s^a e^{-s^2} \left(\frac{s}{\sqrt{t}}\right)^{-a} \rho_D\left(\frac{s}{\sqrt{t}}\right) ds}, \nonumber \\
		& = 2 \lim_{\alpha \to 0} \frac{L}{L} \frac{\int_\R s^{2+a} e^{-s^2} ds}{\int_\R s^a e^{-s^2} ds}, \\
		& = 2 \lim_{t \to \infty} \frac{\frac{1}{2}(1+a)\int_\R s^a e^{-s^2} ds}{\int_\R s^a e^{-s^2} ds} = 1+a. \label{eq:low_energy_limit_spec_dim}
	\end{align}
	For the spectral variance of the same $\rho_D$ we can compute
	\begin{align}
		\lim_{t \to \infty} v_s(t) & = 2 \lim_{t \to \infty} \left(\frac{\int_\R t^2 x^4 e^{-tx^2} \rho_D(x) dx}{\int_\R e^{-tx^2} \rho_D(x) dx} - \left(\frac{\int_\R t x^2 e^{-tx^2} \rho_D(x) dx}{\int_\R e^{-tx^2} \rho_D(x) dx}\right)^2\right), \nonumber \\
		& = 2 \lim_{t \to \infty} \frac{\int_\R s^{4+a} e^{-s^2} \left(\frac{s}{\sqrt{t}}\right)^{-a} \rho_D\left(\frac{s}{\sqrt{t}}\right) ds}{\int_\R s^a e^{-s^2} \left(\frac{s}{\sqrt{t}}\right)^{-a} \rho_D\left(\frac{s}{\sqrt{t}}\right) ds} - \frac{1}{2}(1+a)^2, \nonumber \\
		& = \frac{1}{2}(3+a)(1+a) - \frac{1}{2}(1+a)^2 = 1+a. \label{eq:low_energy_limit_spec_var}
	\end{align}

	Hence if a spectral density has $\rho_D(x) \sim Cx^a$, as $x \to 0$ its spectral dimension and variance are $1 + a$ in the $t \to \infty$ limit.
	If $\rho_D(x) \equiv 0$ in some neighbourhood of zero one can set this as the mass-transformed density of another density $\rho_{D\setminus m}$ that does have a finite limit $\lim_{x \to 0} x^{-a} \rho_{D\setminus m}(x)$; see Equation \ref{eq:mass_transformation}. The effect of such a shift can be computed as in Section \ref{subsec:estimators_non_zero_mass}.

	\subsection{Quadratic and Quartic models with a massless fermion}
	\label{subsec:estimators_massless}

		In the massless case exact solutions for the spectral density can be found and therefore the spectral estimators can also be computed. This is in contrast to the massive case discussed in Section \ref{subsec:estimators_non_zero_mass} where the spectral density is numerically approximated.

		For $m=0$, Equations \ref{eq:one_cut_problem} and \ref{eq:two_cut_problem} simplify to a renormalized version of the fermionless model which can be solved as in \cite{khalkhali2020phase,hessam2022double,d2022numerical,verhoeven2023thesis}. 
		These solutions are, respectively, given by
		\begin{align*}
			& 0 = 12g_4^2 a^8 + 12(\beta+\beta_2)g_4a^4 + 4(\beta+\beta_2)g_2a^2 - (\beta+\beta_2)^2, \\
			& \mu_2 = \frac{2g_4a^6 + g_2 a^4}{\beta+\beta_2-6g_4a^4}, \\
			& \rho_H(x) = \frac{2}{(\beta+\beta_2)\pi} \left(8g_4 x^2 + 4g_4a^2 + 24g_4\mu_2 + 4g_2\right) \sqrt{a^2-x^2},
		\end{align*}
		and
		\begin{align*}
			& a^2 = -\frac{1}{8}\frac{g_2}{g_4} + \sqrt{\frac{\beta+\beta_2}{8g_4}}, \hspace{3em}
			\quad b^2 = -\frac{1}{8}\frac{g_2}{g_4} - \sqrt{\frac{\beta+\beta_2}{8g_4}}, \\
			&\rho_H(x) = \frac{2}{(\beta+\beta_2)\pi} 8g_4|x| \sqrt{(x^2-b^2)(a^2-x^2)},
		\end{align*}
		for the one- and two-cut phases respectively.
		In particular, this allows us to determine that the phase transition occurs on the hypersurface
		\[
			g_2 = -\sqrt{8g_4(\beta+\beta_2)}
		\]
		in $(g_2, g_4, \beta+\beta_2)$-space.
\subsubsection*{Gaussian}
		In the Gaussian case we have $g_4 = 0$ and $g_2 > 0$. 
		In this case there is no phase transition and the solution is always given by
		\begin{align*}
			& a^2 = \frac{\beta+\beta_2}{4g_2}, \\
			& \rho_H(x) = \frac{8g_2}{\pi (\beta + \beta_2)} \sqrt{a^2-x^2}.
		\end{align*}
	In order to compute the spectral dimension and variance from $\rho_H$ we first relate the corresponding normalized heat kernels $k_{D^2}$ and $k_H$. Using the gamma function identity
		\begin{equation*}
			\frac{1}{n!} = \frac{2}{\sqrt{2 \pi} }\int_{0}^{\infty} \frac{(\sqrt{2}s)^{2n}}{(2n)!} e^{-\frac{s^2}{2}} ds
		\end{equation*}
		we can relate the generating functions $k_{D^2}(t)$ to $k_{D}(t)$, as described in \cite{schmidt2017square}.
		Note that
		\begin{align*}
			k_{D^2}(-t) & = \lim_{N\rightarrow \infty} \sum_{n=0} ^{\infty} \langle \frac{1}{N^2} \tr D^{2n} \rangle_{p_N} \frac{t^{n}}{n!}, \nonumber \\
			& = \lim_{N\rightarrow \infty} \frac{2}{\sqrt{2 \pi} }\int_{0}^{\infty} \sum_{n=0} ^{\infty} \langle \frac{1}{N^2} \tr D^{2n} \rangle_{p_N} \frac{(\sqrt{2t}s )^{2n}}{(2n)!}e^{-\frac{s^2}{2}} ds, \nonumber \\
			& = \frac{1}{\sqrt{2 \pi} }\int_{0}^{\infty}\left(k_{D}\left(\sqrt{2t} s\right) + k_{D}\left(-\sqrt{2t} s\right)\right) e^{-\frac{s^2}{2}} ds,
		\end{align*}
		since
		\begin{align*}
			\frac{1}{2}(k_{D}(t) + k_{D}(-t)) & = \lim_{N \rightarrow \infty} \sum_{n=0} ^{\infty} \langle \frac{1}{N^2}\tr D^{2n} \rangle_{p_N} \frac{t^{2n}}{(2n)!}.
		\end{align*}
		
		Now consider the normalized heat kernels for $H$ and $D$. They may be related as follows
		\begin{align*}
			k_{D}(t) & = \lim_{N\rightarrow \infty} \sum_{n=0}^\infty \langle \frac{1}{N^2} \tr D^n \rangle_{p_N} \frac{(-t)^n}{n!}, \\
			& = \lim_{N\rightarrow \infty} \sum_{n=0}^{\infty}\sum_{j=0}^{n} \binom{n}{j} \langle \frac{1}{N} \tr H^{n-j} \rangle_{p_N} \langle \frac{1}{N} \tr H^{j} \rangle_{p_N} \frac{(-t)^n}{n!}, \\
			& = \lim_{N\rightarrow \infty} \left( \sum_{n=0}^{\infty} \langle \frac{1}{N}\tr H^{n} \rangle \frac{(-t)^n}{n!} \right)^2, \\
			& = k_{H}(t)^2.
		\end{align*}
		We used the convolution property of exponential series, the structure of the moments of $D$ in terms of $H$, and Theorem \ref{thm:convergence_of_tracial_observables} to factorize the expectation $\langle \tr H^{n-j} \tr H^j \rangle$. 
In summary, we have the integral formula
		\begin{equation}
			k_{D^2}(-t) = \frac{1}{\sqrt{2 \pi} }\int_{0}^{\infty}\left(k_{H}\left(\sqrt{2t} s\right)^2 + k_{H}\left(-\sqrt{2t} s \right)^2\right)e^{-\frac{s^2}{2}} ds.
			\label{eq:heat_kernel_D2_in_terms_of_H}
		\end{equation}
	
		Using this formula for $k_{D^2}$, it can computed explicitly in the quadratic case. Starting from the known spectral density, we obtain
		\begin{align*}
			k_{H}(t) &=\int e^{-x t} 	\frac{2 g_{2}}{\pi }\sqrt{a^2 - x^2}_{[-a,a]}dx \\
			 &= \frac{2\sqrt{ g_{2}}}{t} I_{1}\left( \frac{t}{\sqrt{ g_{2}}}\right).
		\end{align*}
		This leads to a normalized heat kernel for $D^2$ given by
		\begin{align*}
			k_{D^2}(-t) = \frac{2}{\sqrt{2 \pi} }\int_{0}^{\infty} \left(\frac{\sqrt{ g_{2}}}{\sqrt{2t}s} I_{1}\left( \frac{\sqrt{2t}s}{\sqrt{ g_{2}}}\right)\right)^2 e^{-\frac{s^2}{2}} ds =  \, _2F_2\left(\frac{1}{2},\frac{3}{2};2,3;\frac{4 t}{g_2}\right),
		\end{align*}
		where $ _2F_2$ denotes the generalized hypergeometric function. The spectral dimension can then be computed to be 
		\begin{equation*}
			d_{s}(t) = \frac{t}{g_2}\frac{\, _2F_2\left(\frac{3}{2},\frac{5}{2};3,4;\frac{-4 t}{g_2}\right)}{ \, _2F_2\left(\frac{1}{2},\frac{3}{2};2,3;\frac{-4 t}{g_2}\right)},
		\end{equation*}
		and the spectral variance as
		\begin{equation*}
			v(t) = \frac{t^2}{4 g_2^2}\frac{\left(5 \,
				_2F_2\left(\frac{1}{2},\frac{3}{2};2,3;-\frac{4 t}{g_2}\right)
				\, _2F_2\left(\frac{5}{2},\frac{7}{2};4,5;-\frac{4
					t}{g_2}\right)-2 \,
				_2F_2\left(\frac{3}{2},\frac{5}{2};3,4;-\frac{4
					t}{g_2}\right){}^2\right)}{\,
				_2F_2\left(\frac{1}{2},\frac{3}{2};2,3;-\frac{4
					t}{g_2}\right){}^2}.
		\end{equation*}
		
		Since the spectral density for $D$ is given by the convolution of the semicircular law with itself \cite{khalkhali2022spectral} and this is a smooth function with finite value at 0 the low-energy limit, $t \to \infty$, for both of these is one by Equations \ref{eq:low_energy_limit_spec_dim} and \ref{eq:low_energy_limit_spec_var} for all values of $g_{2}>0$.

		\begin{figure}[h!]
			\centering
			\includegraphics[width=\textwidth]{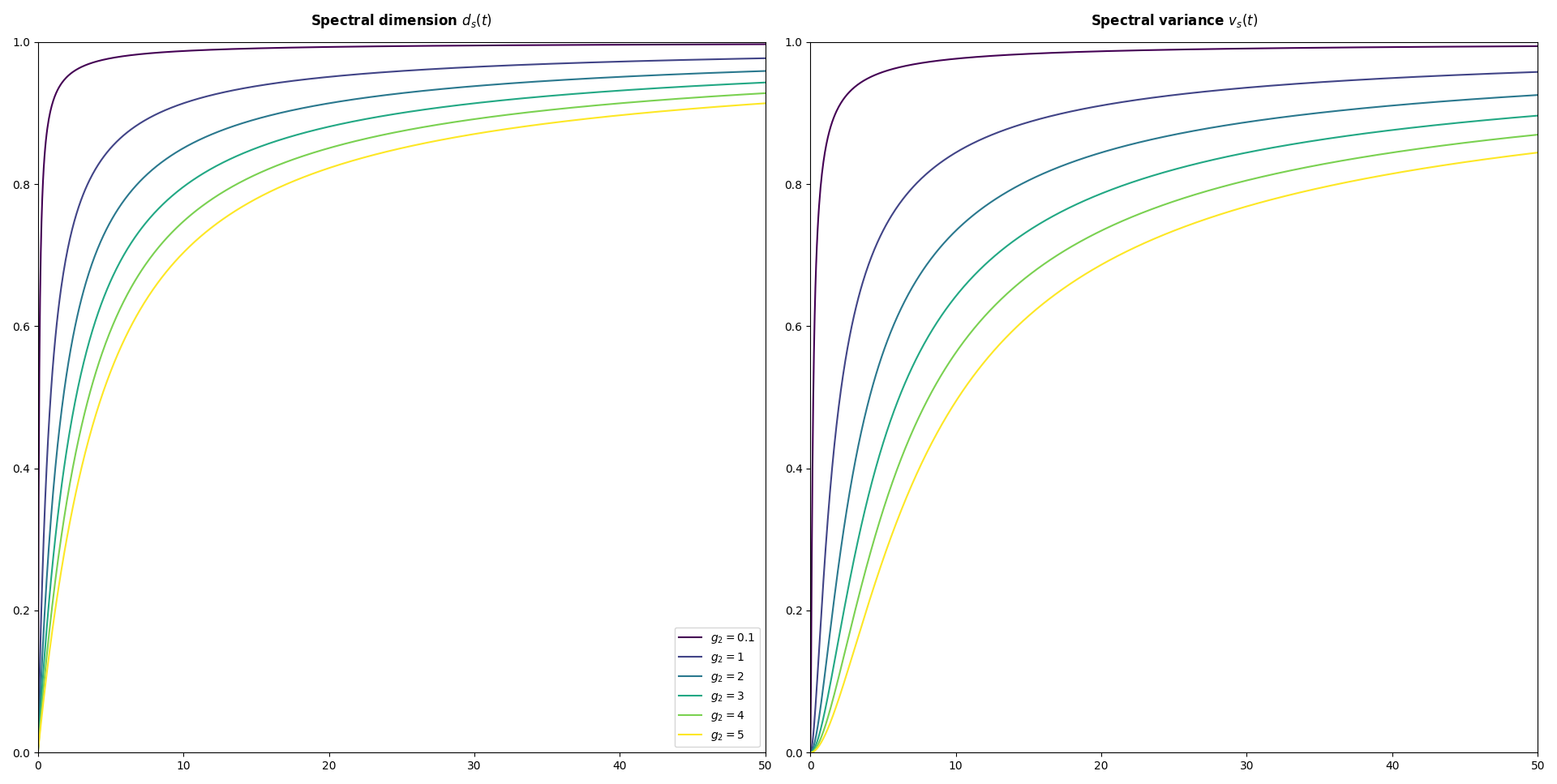}
			\caption{The spectral dimension and variance for the Gaussian massless fermionic $(0, 1)$-fuzzy Dirac ensemble with fixed $\beta = \beta_2 = 2$. The lighter colors correspond to higher values of $g_2$, ranging from 0.1 to 5.}
			\label{fig:estimators_for_quartic}
		\end{figure}

	\subsubsection*{Quartic}


		In the quartic model we proceed in the same way, starting with the one-cut phase.
		Using the known formula for $\rho_H$ we obtain
		\begin{align*}
			k_{H}(t) 
				& = \int_{\text{supp}} e^{-x t}\rho_{H}(x)dx \\
			  & = \frac{2 a}{t^2} \left(\frac{\left(3 a^2 g_{4} \left(a^4 g_{4}-2\right)-2 g_{2}\right) t}{3 a^4 g_{4}-2} I_1(a t) - 6 a g_{4} I_2(a t)\right).
		\end{align*} 

		The expected heat kernel of $D^{2}$ can then be computed using Equation \ref{eq:heat_kernel_D2_in_terms_of_H},
		\begin{align*}
			k_{D^2}(t) &= \int_{0}^{\infty}\frac{2 a^2 e^{-\frac{s^2}{2}}}{
				\sqrt{2 
			\pi} t^2 s^4}  \left(-\frac{\sqrt{2t} s \left(3 a^2 g_{4} \left(a^4 g_{4}-2\right)-2
				g_{2}\right)  J_1\left(\sqrt{2} a s \sqrt{t}\right)}{3 a^4 g_{4}-2}+6 a g_{4}
			J_2\left(\sqrt{2} a s \sqrt{t}\right)\right){}^2 ds
		\end{align*}

		The authors were unable to compute a closed-form expression for the above integral. However, the above form lends itself nicely to numerical integration, allowing us to produce the plots in Figure \ref{fig:quartic_onecut_estimators_massless}. It is clear from the figure that the low-energy limit of the spectral dimension is one. This can be seen to be the limit of both the spectral dimension of Equation \ref{eq:low_energy_limit_spec_dim} and the spectral variance of Eq. \ref{eq:low_energy_limit_spec_var}.

\begin{figure}[h!]
	\centering
	\includegraphics[width=\textwidth]{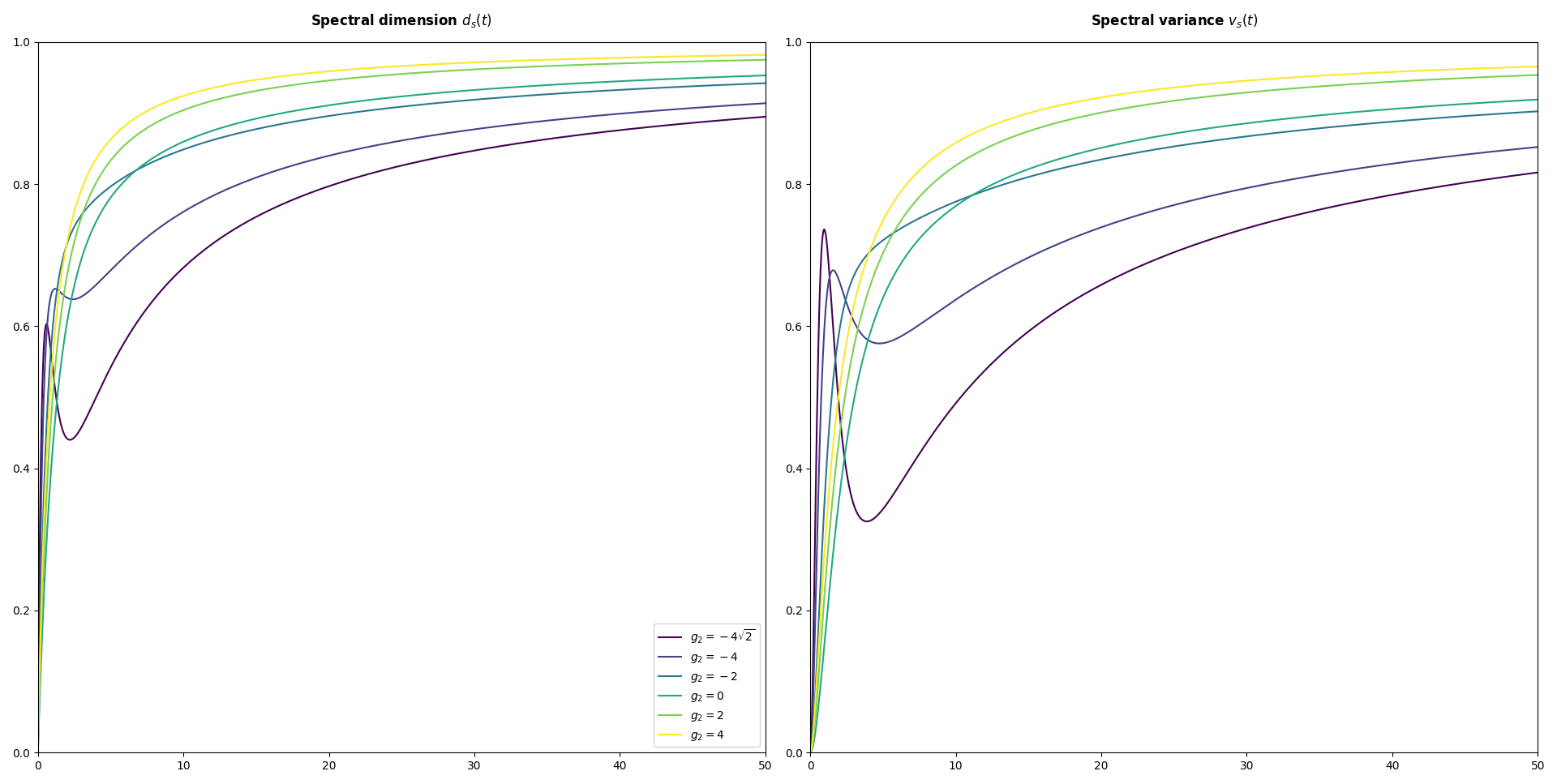}
	\caption{The spectral dimension and variance of the quartic massless $(0, 1)$-fuzzy Dirac ensemble in the one cut phase for $g_{4} = 1$ and $\beta = \beta_{2} = 2$. Darker colors correspond to lower values of $g_2$, with the lowest value at the phase transition.}
	\label{fig:quartic_onecut_estimators_massless}
\end{figure}

\begin{figure}[h!]
	\centering
	\includegraphics[width=\textwidth]{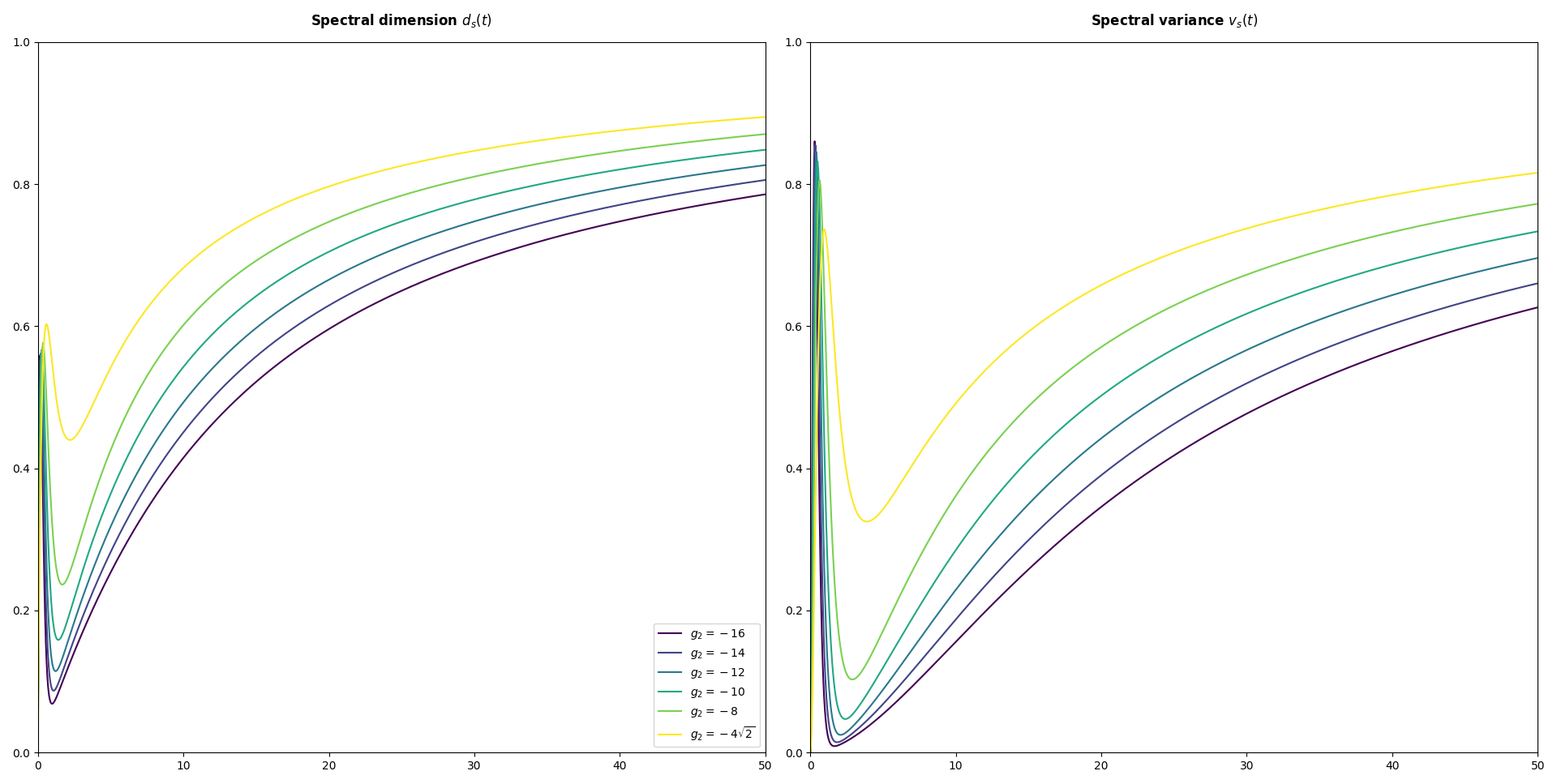}
	\caption{The spectral dimension and variance of the quartic massless $(0, 1)$-fuzzy Dirac ensemble in the two cut phase for $g_{4} = 1$ and $\beta = \beta_{2} = 2$. Darker colors correspond to lower values of $g_2$, with the highest value at the phase transition.}
	\label{fig:quartic_twocut_estimators_massless}
\end{figure}

		Let us further investigate the nature of the spectral phase transition, i.e. the transition occurs between the validity of Equation \ref{eq:one_cut_problem} and Equation \ref{eq:two_cut_problem}.
		In the fermionless model the phase transition happens, qualitatively, when the balance between the eigenvalue repulsion from the Vandermonde term
		$$
			-\frac{\beta}{4} \sum_{i \neq j} \log\left((\lambda_i - \lambda_j)^2\right)
		$$
		and the confining potential
		$$
			2g_4\left(m^2 + (\lambda_i-\lambda_j)^2\right)^2 + 2g_2\left(m^2 + (\lambda_i-\lambda_j)^2\right)
		$$
		shifts.
		
		For negative values of $g_2$ the potential has a double well, as can be seen in Figure \ref{fig:double_well_potential}.
		If this double well is deep enough compared to the strength of the repulsion, the eigenvalues will cluster in the wells leading to a spectral density with disconnected support.
		If the wells are not deep enough compared to the repulsive effect, the spectral density will be nonzero on a connected interval.
		With the addition of a massless fermion the fermionic action simply adds to the Vandermonde repulsion, changing the coupling constant of the repulsion from $\beta$ to $\beta + \beta_2$.

		\begin{figure}[h!]
			\centering
			\includegraphics[width=0.6\textwidth]{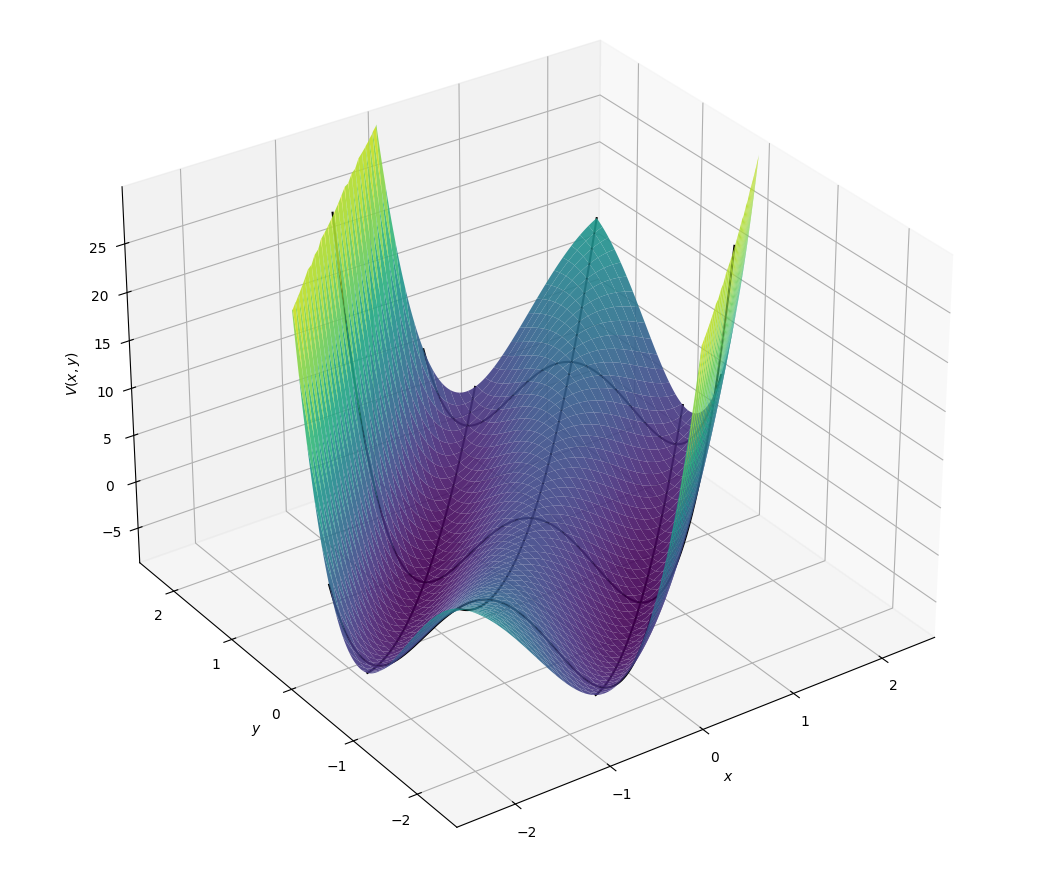}
			\caption{The polynomial part of the interaction potential of the fermionic quartic $(0, 1)$-fuzzy geometry, see Proposition \ref{prop:our_model_satisfies_assumptions}, with $g_4 = 1$, $g_2 = -4$, $m = 0$. The black lines are added for emphasis of the shape in the directions with $x-y$ and $x+y$ constant. Along lines of constant $x-y$ the potential is quadratic, along lines of constant $x+y$ the potential is quartic and for $g_2 < 0$ has a double-well profile.}
			\label{fig:double_well_potential}
		\end{figure}

		In order to see how the model changes at this phase transition we can use the formulas for $\rho_H$ to compute the moments of $H$ in terms of the coupling constants.
		In the single cut phase we find
		\[
			\mu_{2n} := \int x^{2n}\rho_{H}(x)dx = \frac{a^2}{\beta+\beta_2} \left(\frac{a}{2}\right)^{2n} \left(2g_4 a^2 C_{n+1} + (4g_4 a^2 + 24\mu_2 + 4g_2) C_{n} \right),
		\]
		where $C_{n}$ denotes the $n$-th Catalan number. In the double-cut phase we can obtain a recursive relation for the moments by expressing the relevant integrals as hypergeometric functions and using Gauss's contiguous relations (see for example \cite[Ch. 9]{zwillinger2007table}).
		\begin{align*}
			\mu_0 & = 1, \\
			\mu_2 & = -\frac{1}{8}\frac{g_2}{g_4}, \\
			\mu_{2(n+2)} & = - \frac{1}{4} \frac{2n+5}{2n+8} \frac{g_2}{g_4} \mu_{2(n+1)} - \frac{n+1}{n+4} \left(\frac{1}{64} \frac{g_2^2}{g_4^2} - \frac{1}{8} \frac{\beta + \beta_2}{g_4}\right) \mu_{2n}.
		\end{align*}
		From here the moments for the Dirac operator can be computed using the relation $\rho_D = \rho_H * \rho_H$, so that
		\[
			\mu_{2n}^D =  \sum_{j = 0}^{2n} \binom{2n}{j} (-1)^{j} \mu_{2n-j} \mu_{j}.
		\]
		
		When considering the moments as a function of the coupling constants, the spectral phase transition is second order.
		For example, we have for the second moment of $H$, $\mu_2$, at fixed $g_4$ and $\beta + \beta_2$, that
		\begin{align*}
			\lim_{g_2 \uparrow g_{2, crit}} \mu_2 & = \sqrt{\frac{\beta + \beta_2}{8g_4}} = \lim_{g_2 \downarrow g_{2, crit}} \mu_2, \\
			\lim_{g_2 \uparrow g_{2, crit}} \frac{\partial }{\partial g_2} \mu_2 & = -\frac{1}{8g_4} = \lim_{g_2 \downarrow g_{2, crit}} \frac{\partial }{\partial g_2} \mu_2, \\
			\lim_{g_2 \uparrow g_{2, crit}} \frac{\partial^2 }{\partial g_2 ^2} \mu_2 & = 0 \neq \frac{1}{64} \frac{1}{\sqrt{2(\beta+\beta_2)g_4^3}} = \lim_{g_2 \downarrow g_{2, crit}} \frac{\partial^2 }{\partial g_2 ^2} \mu_2.
		\end{align*}
		This discontinuity in the second derivative is also present when expressing higher order moments, moments of $D$ and the outer boundary of the support in terms of any of the coupling constants $g_2$, $g_4$, $\beta + \beta_2$.

	\subsection{Spectral estimators for non-zero mass}
		\label{subsec:estimators_non_zero_mass}

		\subsubsection*{Fermionic effects on the equations for $\rho$}

		The addition of a massive fermion by means of a finite spectral triple as in Equation \ref{eq:finite_spectral_triple} changes the spectral density.
		The changes can be split into three different mechanisms: a shift in coupling constants, the appearance of the integral kernel $K_{\Sigma, m}$ in the equations for $\rho_H$, and a different relation between $\rho_H$ and $\rho_D$.

		Let us start with this first effect. Expanding the quartic term in $D$ yields
		\[
			\Tr\left((D_{fuzzy} \otimes \sigma_1 + m \otimes \sigma_2)^4 \right) = 2\Tr\left(D_{fuzzy}^4\right) + 2 m^2 \Tr\left(D_{fuzzy}^2\right) + 2 m^4,
		\]
		so the quadratic coupling constant $g_2$ is shifted to $g_2 + 2m^2 g_4$.
		The constant term $2m^4$ does not depend on $D$ and therefore represents a constant shift in energy that does not affect the probability distribution.
		Similarly the quadratic term also gives an additional constant shift in energy by $2m^2g_2$.
		Since this renormalization effect quickly dominates the influence of changing $m$ we usually keep $g_2' = g_2 + 2m^2 g_4$ constant instead of $g_2$ when varying $m$.

		The second effect is the integral kernel appearing in the Fredholm integral equations \ref{subeq:one_cut_problem_rho} and \ref{subeq:two_cut_problem_rho}, which is absent in the model with no fermions.
		These equations can be re-expressed relative to the function $\sqrt{s}_+(x) := \frac{1}{i} \lim_{\epsilon \to 0} \sqrt{s}(x+i\epsilon)$. 
		Letting $\phi(x) = \frac{\rho(x)}{\sqrt{s}_+(x)}$ we obtain the equation
		\begin{equation}
			\phi(x) = p(x) + \int_{\Sigma} \widetilde{K}(x, y) \phi(y) d y
			\label{eq:equation_for_phi},
		\end{equation}
		where
		\[
			\widetilde{K}(x, y) = \frac{\beta_2}{\pi \beta}\Re\left( \frac{\sqrt{s}_+(y)}{\sqrt{s}(y+im)} \frac{1}{y-x + im}\right),
		\]
		and $p(x)$ is some polynomial determined by the coupling constants and phase of the model.
		We have been unable to solve these integral equations, but their solutions can be approximated numerically.
		For this the formulation in terms of $\phi$ is more stable since $\widetilde{K}$ is better behaved near the edges of $\Sigma$, especially for small masses.
		The effect of the integral kernel on the solutions $\phi$ and $\rho$ can be seen in Figure \ref{fig:effect_of_fermion_kernel_onecut} for $\Sigma = [-2, 2]$ and $p(x) = 1$.
		
		The limiting behaviour for $m \to \infty$ and $m \to 0$ can be explained from the limiting behaviour of $K$ and $\tilde{K}$.
		In the large mass limit, $m \to \infty$, the kernels tend to 0 and the solutions tend to the solutions from the model without fermions.
		This is of course predicated on keeping $g_2'$ constant, otherwise $g_2$ will also tend to infinity making $\rho$ tend towards a delta mass at 0.
		For the small mass limit we have
		\begin{align*}
			\lim_{m \to 0} \widetilde{K}(x, y) & = \lim_{m \to 0} \frac{\beta_2}{\pi \beta}\Re\left( \frac{\sqrt{s}_+(y)}{\sqrt{s}(y+im)} \frac{1}{y-x + im}\right), \\
			& = \frac{\beta_2}{\pi \beta} \lim_{m \to 0} \Re\left(\frac{1}{i} \frac{1}{y-x+im}\right), \\
			& = \frac{\beta_2}{\pi \beta} \lim_{m \to 0} -\frac{m}{(y-x)^2 + m^2}.
		\end{align*}
		Since $\frac{1}{\pi} \frac{m}{(y-x)^2 + m^2}$ is a nascent delta function as $m \to 0$ this gives $\widetilde{K}(x, \cdot) \to -\delta_x$ as distribution. 
		So in the $m \to 0$ limit Equation \ref{eq:equation_for_phi} becomes
		\[
			\phi(x) = p(x) - \frac{\beta_2}{\beta}\phi(x).
		\]

		\begin{figure}[h!]
			\centering
			\includegraphics[width=\textwidth]{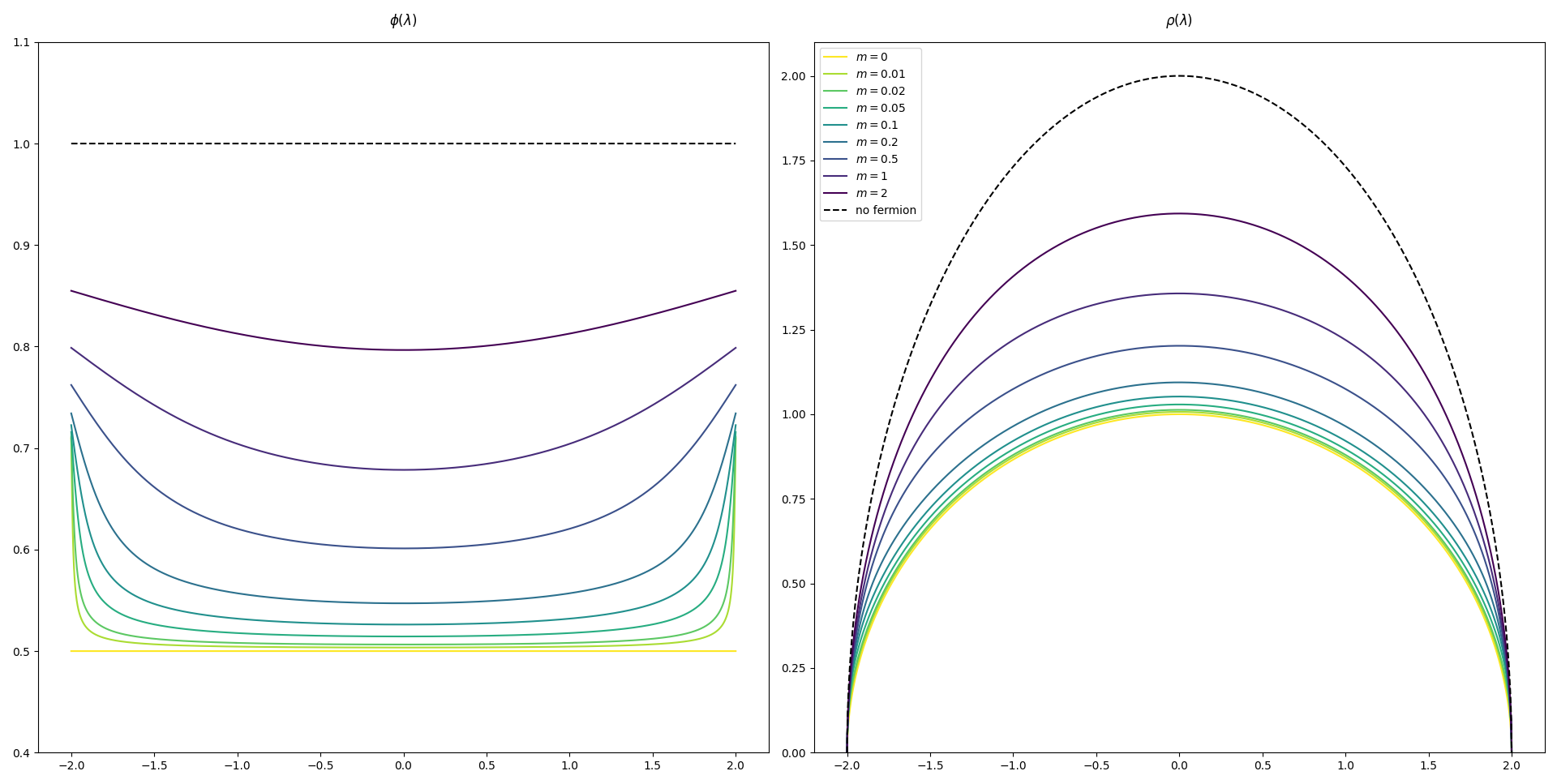}
			\caption{On the left the solutions $\phi$ of Equation \ref{eq:equation_for_phi} for $p(x)=1$ and $\Sigma = [-2, 2]$, with $\frac{\beta_2}{\beta}=1$ and various values of $m$ with lighter colors corresponding to lower mass. On the right the corresponding solutions $\rho$ to Equation \ref{subeq:one_cut_problem_rho}. Note that the solution $\rho$ to the full Equation \ref{eq:one_cut_problem} would be normalized, so the presence of the fermions affects which support $\Sigma$ gives the full solution. Observe that in the limit $m \to 0$ the solution $\phi$ converges to $\frac{1}{2}p$, while in the $m \to \infty$ limit the solution $\phi$ converges to $p$.}
			\label{fig:effect_of_fermion_kernel_onecut}
		\end{figure}

		
		Besides changing the coupling constants and the spectral density of $H$, the third effect of massive fermions is the relation between the spectrum of $H$ and that of $D = D_{fuzzy} \otimes 1 + 1 \otimes D_{finite}$, as seen between Equations \ref{eq:spectrum_of_fuzzy_D_from_H} and \ref{eq:spectrum_of_D_with_mass_from_H}.
		This shifts the eigenvalues of $D^2$ by a value of $m^2$ relative to those of the bare $D_{fuzzy}^2 \otimes 1_2$.

		We can express the spectral density for $D_{fuzzy} \otimes 1 + 1 \otimes m$ in terms of the spectral density for $D_{fuzzy}$ as follows:
		\begin{equation}
			\label{eq:mass_transformation}
			\rho_{D_{fuzzy} + m}(\lambda) = \left\{ \begin{array}{ll}
				0, & -m < \lambda < m \\
				\frac{\lambda}{\sqrt{\lambda^2 - m^2}} \rho_{D_{fuzzy}}\left(\sqrt{\lambda^2 - m^2}\right), & |\lambda| > m
			\end{array} \right.
		\end{equation}

	\subsubsection*{Effects on the phase transition}

		To understand how the presence of fermions in a fuzzy geometry affects the geometry, we first consider the phase transition.
		As discussed in Section \ref{subsec:estimators_massless}, the phase transition depends on a balance between the repulsion due to the Vandermonde and fermionic terms and the depth of the double-well potential.
		A massless fermion coincides with the Vandermonde repulsion, but for a massive fermion the energy of two eigenvalues at distance $r$ is given by $\log\left(m^2 + r^2\right)$	instead of the Vandermonde $\log\left(r^2\right)$.

		The addition of the mass in this energy has some interesting effects. 
		In Figure \ref{fig:phasetransition_and_mass} the equilibrium measures for $H$ and the Dirac operator $D$ are plotted for various values of the coupling constant $g_2$ and various masses $m$.
		For models with mass 0 the phase transition is at $g_2 = -2\sqrt{2g_4(\beta + \beta_2)}$.

		The general trend, visible most clearly in the densities for $g_2' = -5, -6$ in Figure \ref{fig:phasetransition_and_mass}, is that the addition of a fermion increases the eigenvalue repulsion, pushing the model towards the one-cut phase.
		An interesting exception can be seen in Figure \ref{fig:anomalous_force}, where as the mass increases the model moves from the one-cut phase into the two-cut phase and then back into the one-cut phase.
		This is surprising since the presence of a fermion increases the repulsive force between eigenvalues and would be expected to push the model always closer to the single cut phase.

		\begin{figure}[h!]
			\centering
			\includegraphics[width=\textwidth]{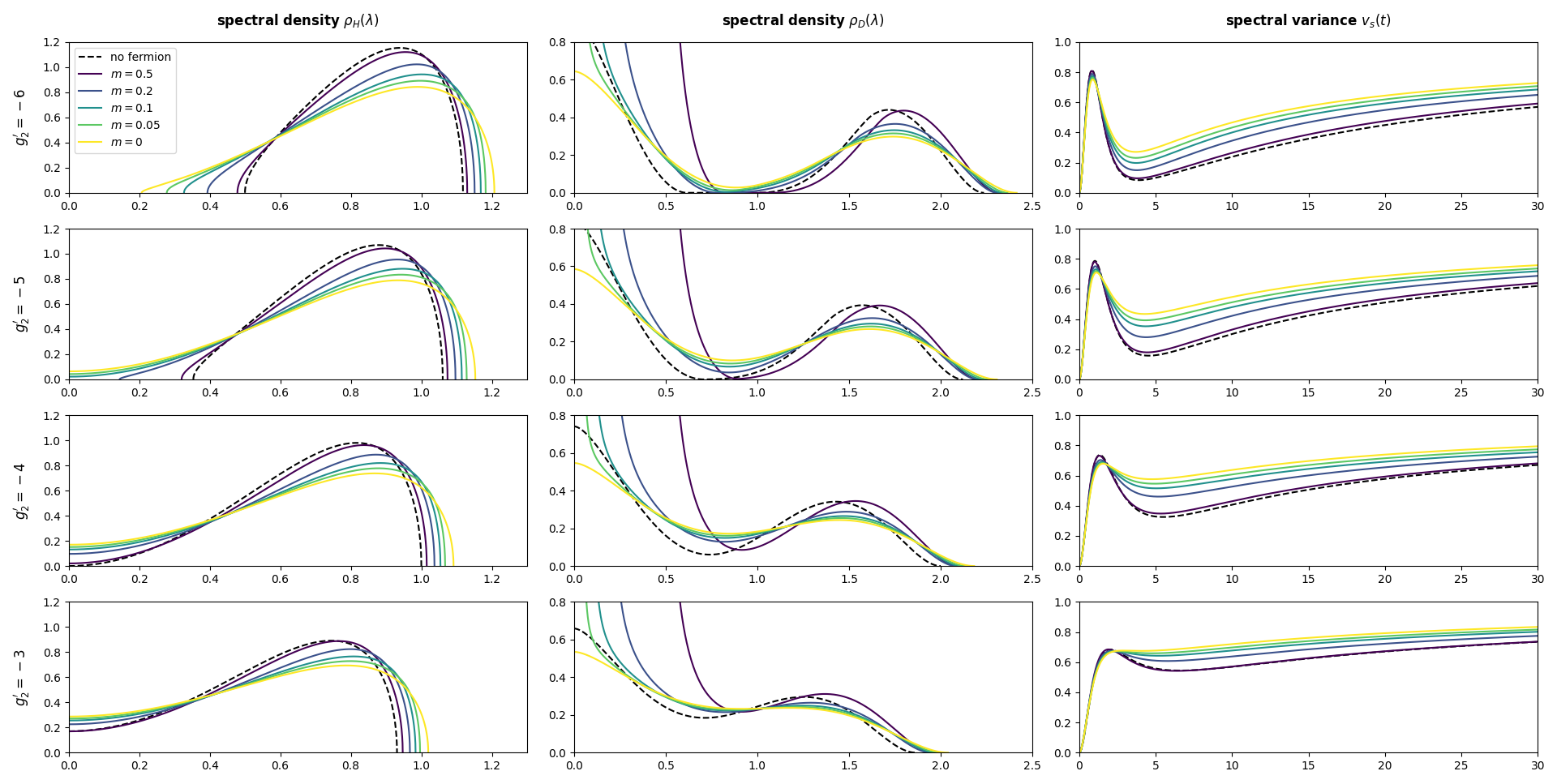}
			\caption{The spectral densities of $H$ and $D$ as well as the spectral variance for various values of the coupling constants $g_2'$ and $m$. The other coupling constants are fixed values at $g_4 = 1$, $\beta = \beta_2 = 2$. Only the range $x > 0$ is shown to allow for greater detail, the spectral densities are symmetric. The phase transition for the model without fermion ($\beta_2 = 0$) and these coupling constants lies at $g_2 = -4$; the phase transition for the $m = 0$ model lies at $g_2 = -2\sqrt{8} \approx -5.66$.
			The limit $\lim_{t \to \infty} v_s(t)$ is still 1 for all of these ensembles, but the convergence is very slow.}
			\label{fig:phasetransition_and_mass}
		\end{figure}
		
		To understand this effect consider the forces acting on a single eigenvalue in one of the wells.
		These are the forces due to the fermionless model, the Vandermonde repulsion, and the force due to the potential, plus the fermion repulsion.
		The fermion repulsion due to $\log\left(m^2 + r^2\right)$ has a maximum at $r = m$, so if the distance between the wells is approximately $m$ the force due to the eigenvalues in the other well is stronger than the force due to eigenvalues in the same well.
		This causes eigenvalues to be more repelled by the other well than their direct neighbours, pushing the model closer to the double cut phase.

		\begin{figure}[h!]
			\centering
			\includegraphics[width=\textwidth]{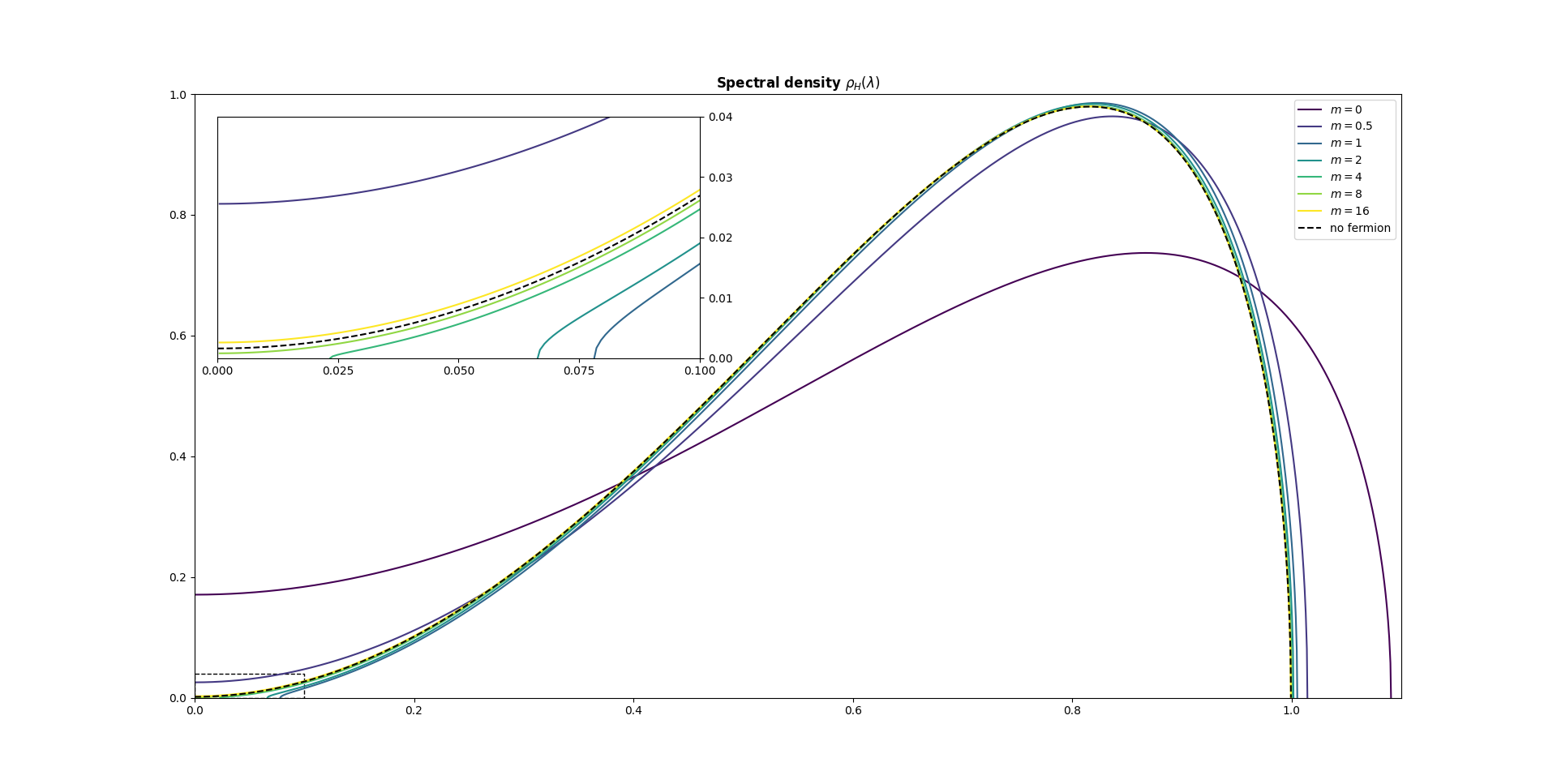}
			\caption{The positive part of the spectral density of $H$ for $g_2' = -3.99$, $g_4 = 1$, $\beta = \beta_2 = 2$. Note that for $m = 0$, $m = 0.5$, $m= 8$, and $m = 16$ the model is in the one-cut phase, while for $m = 1$, $m=2$, and $m=4$ the model is in the two-cut phase. This is due to the repulsive effect of the fermions being stronger between the two wells than within a well if the mass is approximately the separation between the wells. }
			\label{fig:anomalous_force}
		\end{figure}

	\subsubsection*{Effects on observables}

		Beside changing the equations for $\rho_H$ as observed in the previous sections, the presence of a massive fermion also affects the relation between $\rho_H$ and $\rho_D$, as described by Equation \ref{eq:mass_transformation}.

		The effect of this transformation on the spectral dimension is
		\begin{align*}
			d^m_s(t) & = 2t \frac{\int (m^2 + \lambda^2) e^{-(m^2 + \lambda^2)t} \rho_D(\lambda) d\lambda}{\int e^{-(m^2 + \lambda^2)t} \rho_D(\lambda) d\lambda}, \\
			& = 2t \left(m^2 + \frac{\int \lambda^2 e^{-\lambda^2t} \rho_D(\lambda) d\lambda}{\int e^{-\lambda^2t} \rho_D(\lambda) d\lambda}\right) = 2tm^2 + d^0_s(t)
		\end{align*}
		while the spectral variance remains unaffected
		\begin{align*}
			v^m_s(t) & = 2t^2 \left(\frac{\int (m^2 + \lambda^2)^2 e^{-(m^2 + \lambda^2)t} \rho_D(\lambda) d\lambda}{\int e^{-(m^2 + \lambda^2)t} \rho_D(\lambda) d\lambda} - \left(\frac{\int (m^2 + \lambda^2) e^{-(m^2 + \lambda^2)t} \rho_D(\lambda) d\lambda}{\int e^{-(m^2 + \lambda^2)t} \rho_D(\lambda) d\lambda}\right)^2\right), \\
			& = 2t^2 \left(m^4 + 2m^2 \frac{\int \lambda^2 e^{-\lambda^2t} \rho_D(\lambda) d\lambda}{\int e^{-\lambda^2t} \rho_D(\lambda) d\lambda} + \frac{\int \lambda^4 e^{-\lambda^2t} \rho_D(\lambda) d\lambda}{\int e^{-\lambda^2t} \rho_D(\lambda) d\lambda} - \left(m^2 + \frac{\int \lambda^2 e^{-\lambda^2t} \rho_D(\lambda) d\lambda}{\int e^{-\lambda^2t} \rho_D(\lambda) d\lambda}\right)^2\right), \\
			& = 2t^2 \left(\frac{\int \lambda^4 e^{-\lambda^2t} \rho_D(\lambda) d\lambda}{\int e^{-\lambda^2t} \rho_D(\lambda) d\lambda} - \left(\frac{\int \lambda^2 e^{-\lambda^2t} \rho_D(\lambda) d\lambda}{\int e^{-\lambda^2t} \rho_D(\lambda) d\lambda}\right)^2\right) = v^0_s(t).
		\end{align*}

		The behaviour of the spectral variance as the mass varies at various $g_2'$ can be seen in the third column of Figure \ref{fig:phasetransition_and_mass}.
		The influence of $g_2$ and $g_4$ at various fixed masses is shown in Figures \ref{fig:influence_of_g4} and \ref{fig:influence_of_g2}. 
		Finally, in Figure \ref{fig:influence_of_beta2}, the strength of the fermionic coupling constant $\beta_2$ is varied. This figure shows that, as the mass increases, the fermionic repulsion (or attraction for $\beta_2 < 0$) becomes less and less significant.

		\begin{figure}[h!]
			\centering
			\includegraphics[width=\textwidth]{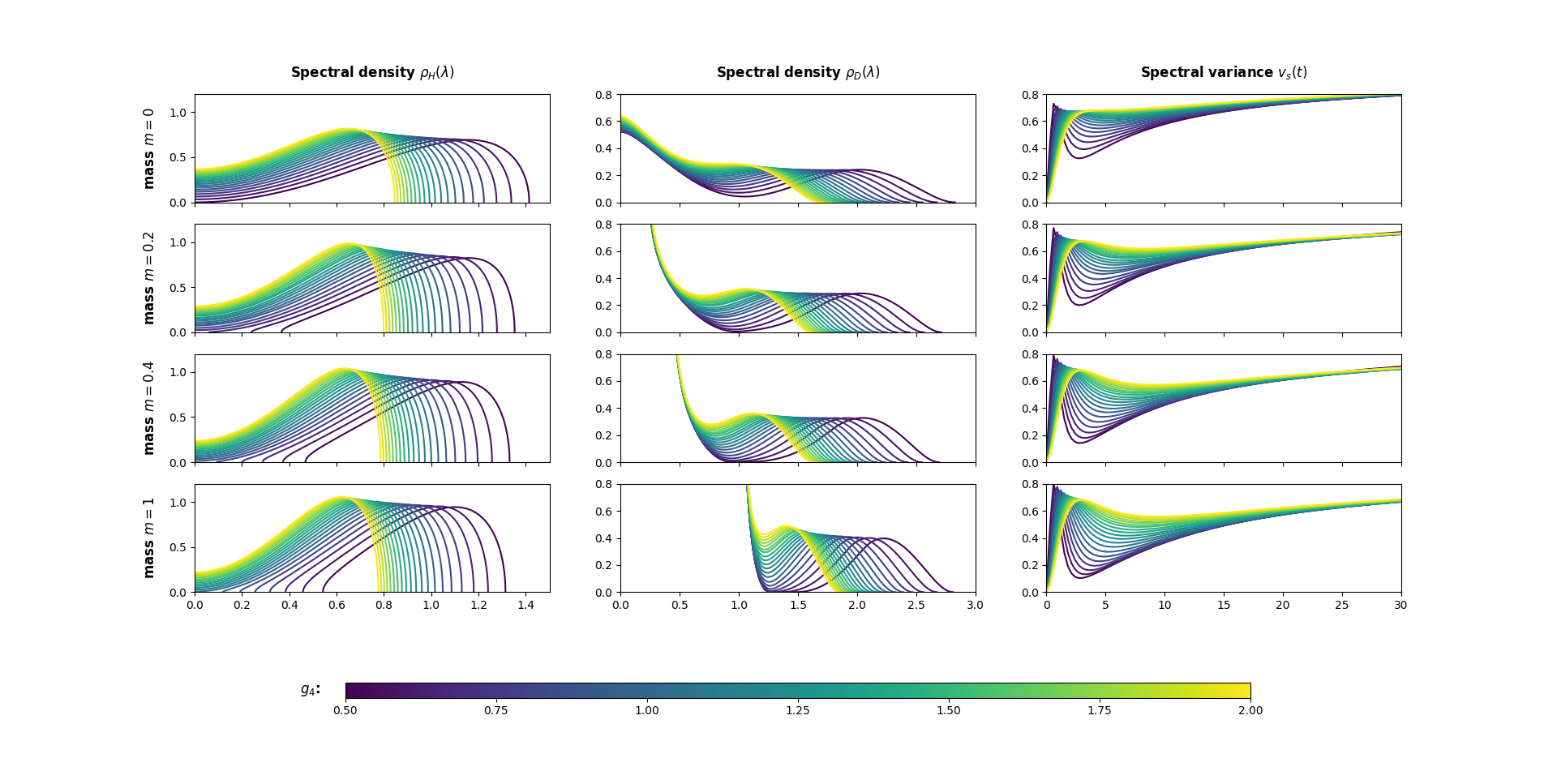}
			\caption{The spectral densities of $H$ and $D$ together with the spectral variance as functions of the coupling constant $g_4$ at various masses. The other coupling constants are kept fixed at $g_2' = -4$, $\beta = \beta_2 = 2$. Large $g_4$ (brighter color) correspond to more concentrated densities as the confining effect of $g_4 \Tr(D^4)$ increases. The model also tends to the one-cut phase for increasing $g_4$ as the double well caused by $g_2'$ becomes relatively less significant.}
			\label{fig:influence_of_g4}
		\end{figure}
		
		\begin{figure}[h]
			\centering
			\includegraphics[width=\textwidth]{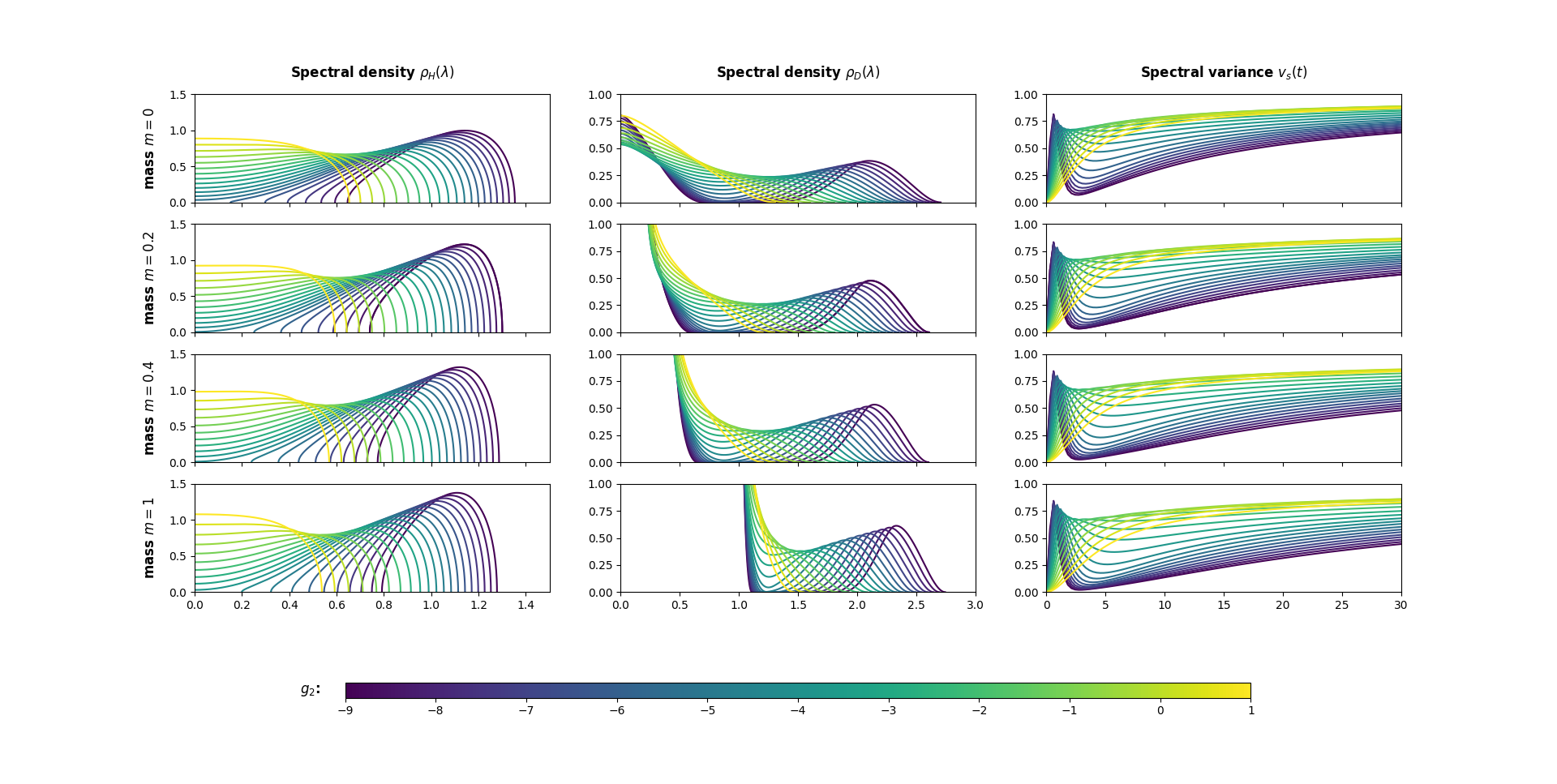}
			\caption{The spectral densities of $H$ and $D$ together with the spectral variance as functions of the coupling constant $g_2$ at various masses. The other coupling constants are kept fixed at $g_4 = 1$, $\beta = \beta_2 = 2$. Very negative values of $g_2'$ (darker color) correspond to two-cut models, as the double well becomes more pronounced. This also lengthens the gap in the Dirac density, and thus the dip in the spectral variance, as the spacing between the phases increases.}
			\label{fig:influence_of_g2}
		\end{figure}

		\begin{figure}[h!]
			\centering
			\includegraphics[width=\textwidth]{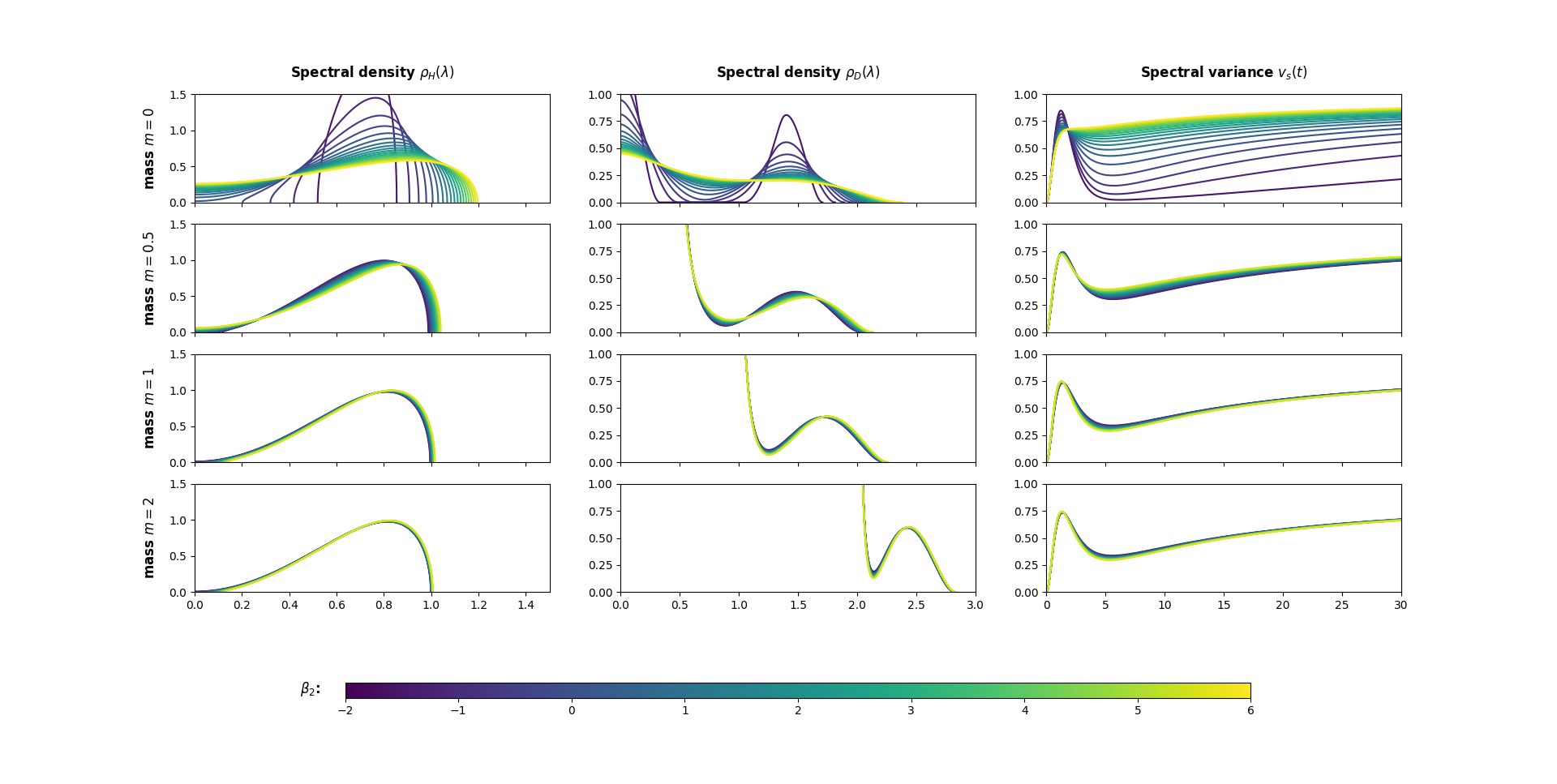}
			\caption{The spectral densities of $H$ and $D$ together with the spectral variance as functions of the fermionic coupling constant $\beta_2$ at various masses. The other coupling constants are kept fixed at $g_4 = 1$, $g_2' = -4$, $\beta = 2$. Since $\beta_2$ corresponds to the strength of the fermionic repulsion higher (brighter) values correspond to one-cut models with a flatter Dirac density. Note that as the mass increases the effect of the fermionic repulsion quickly becomes insignificant.
			Note that for $m = 0$ the case $\beta_2 = -2$ is not plotted. In that model there is no eigenvalue repulsion so that the spectral density is a Dirac mass at the minimum of the potential.}
			\label{fig:influence_of_beta2}
		\end{figure}

\section{Discussion and outlook}
\label{sec:discussion}
	
	In this paper we have restricted ourselves to models with a single fermion.
	Having multiple fermions is certainly possible and does not require new methods than those presented, and they do not significantly change the behaviour of the model.
	When adding multiple fermions only the eigenvalues of the corresponding mass matrix, i.e. the Dirac operator of the finite spectral triple representing fermion space, matter.
	Letting $m_1, \ldots, m_n$ be these eigenvalues, the fermionic action will be
	\[
		\frac{\beta_2}{4} \sum_{i=1}^n \sum_{j, k = 1}^N \log\left(m_i^2 + (\lambda_j-\lambda_k)^2\right),
	\]
	and the integral kernel $K_{\Sigma, m}$ in Equations \ref{eq:one_cut_problem}and \ref{eq:two_cut_problem} will be replaced by $K_{\Sigma, m_1} + \ldots + K_{\Sigma, m_n}$, where $m_1, \ldots, m_n$ are the eigenvalues of the mass matrix.
	This does not produce qualitatively different behaviour than can already be observed in the single fermion case, so for simplicity we have used a single fermion throughout.

	Adding multiple fermions is likely to be interesting in possible extensions of these models that include a random mass matrix or a different fermionic action. Additionally, adding a number of massless fermions generates examples of $\beta$-ensembles with arbitrarily high integer $\beta$.
	The dimension of the fermion space multiplies the coupling constants of the action since
	\[
		\Tr\left(D_{\text{geom}} \otimes 1_{\text{fermion}}\right) = \dim\left(H_{\text{fermion}}\middle) \cdot \Tr\middle(D_{\text{geom}}\right).
	\]
	So an $n$-dimensional fermion space with $D_{\text{fermion}} = 0$ corresponds to a Coulomb gas model at inverse temperature $\beta = n$.
		
	We have also chosen to investigate our Dirac ensemble through the Coulomb gas method; this is certainly not the only available method.
	Recall that a formal matrix integral is a formal power series constructed by expanding the non-Gaussian part of the action in a matrix integral
	\[
		Z = \int_{\mathcal{H}_N} e^{-\frac{1}{2}\Tr(H^2) + t\Tr(V(H))} dH
	\]
	in a power series in $t$, and interchanging the order of summation and integration.
	If $V$ is polynomial these formal matrix integrals have graphical interpretations \cite{eynardRandomMatrices}.
	The matrix ensemble for the fermionic $(0, 1)$-fuzzy Dirac ensemble has a corresponding formal matrix integral, in $g_4$ and $\beta_2$.
	It would be interesting to find the graphical interpretation of the fermionic ensemble and compare it to the purely fuzzy ensembles examined in \cite{khalkhali2022spectral}, especially since these graphical tools are also available for multi-matrix ensembles.

	Additionally, the large $N$ behaviour of observables of such matrix ensembles often coincides with that of their convergent counterparts \cite{guionnet2005combinatorial,figalli2016universality}.
	For many Dirac ensembles this can be proved \cite{hessam2022bootstrapping,hessam2022double}. 
	In particular, this correspondence was recently used to find the large $N$ limits of the moments and the partition function for two-matrix Dirac ensembles in \cite{khalkhali2023coloured}.
	Such models can further be viewed as sums over random combinatorial maps; these random maps can also be assigned notions of dimension \cite{gwynne2021random}.
	Since these random maps provide a discretization of Liouville Quantum Gravity, comparing the spectral dimension of fuzzy Dirac ensembles to those of the corresponding random maps can shed further light on a possible duality between discretization and fuzzification of geometries.

\bibliographystyle{abbrv}

\begin{thebibliography}{10}

\bibitem{ambjorn2005spectral}
J.~Ambj{\o}rn, J.~Jurkiewicz, and R.~Loll.
\newblock The spectral dimension of the universe is scale dependent.
\newblock {\em Physical review letters}, 95(17):171301, 2005.

\bibitem{azarfar2019random}
S.~Azarfar and M.~Khalkhali.
\newblock Random finite noncommutative geometries and topological recursion.
\newblock {\em Annales de l’Institut Henri Poincar{\'e} D}, 2024.

\bibitem{Barrett_FermionDoubling}
J.~W. Barrett.
\newblock Lorentzian version of the noncommutative geometry of the standard
  model of particle physics.
\newblock {\em Journal of mathematical physics}, 48(1):012303, 2007.

\bibitem{barrett2015matrix}
J.~W. Barrett.
\newblock Matrix geometries and fuzzy spaces as finite spectral triples.
\newblock {\em Journal of Mathematical Physics}, 56(8):082301, 2015.

\bibitem{barrett2024fermion}
J.~W. Barrett.
\newblock Fermion integrals for finite spectral triples.
\newblock {\em arXiv preprint arXiv:2403.18428}, 2024.

\bibitem{barrett2019spectral}
J.~W. Barrett, P.~Druce, and L.~Glaser.
\newblock Spectral estimators for finite non-commutative geometries.
\newblock {\em Journal of Physics A: Mathematical and Theoretical},
  52(27):275203, 2019.

\bibitem{barrett2019torus}
J.~W. Barrett and J.~Gaunt.
\newblock Finite spectral triples for the fuzzy torus.
\newblock {\em arXiv preprint arXiv:1908.06796}, 2019.

\bibitem{barrett2016monte}
J.~W. Barrett and L.~Glaser.
\newblock Monte carlo simulations of random non-commutative geometries.
\newblock {\em Journal of Physics A: Mathematical and Theoretical},
  49(24):245001, 2016.

\bibitem{berline2003heat}
N.~Berline, E.~Getzler, and M.~Vergne.
\newblock {\em Heat kernels and Dirac operators}.
\newblock Springer Science \& Business Media, 2003.

\bibitem{bourgade2014edge}
P.~Bourgade, L.~Erd{\"o}s, and H.-T. Yau.
\newblock Edge universality of beta ensembles.
\newblock {\em Communications in Mathematical Physics}, 332(1):261--353, 2014.

\bibitem{branahl2022scalar}
J.~Branahl, A.~Hock, H.~Grosse, and R.~Wulkenhaar.
\newblock From scalar fields on quantum spaces to blobbed topological
  recursion.
\newblock {\em Journal of Physics A: Mathematical and Theoretical},
  55(42):423001, 2022.

\bibitem{carlip2017dimension}
S.~Carlip.
\newblock Dimension and dimensional reduction in quantum gravity.
\newblock {\em Classical and Quantum Gravity}, 34(19):193001, 2017.

\bibitem{chamseddine1997spectral}
A.~H. Chamseddine and A.~Connes.
\newblock The spectral action principle.
\newblock {\em Communications in Mathematical Physics}, 186(3):731--750, 1997.

\bibitem{ChamsedinneConnesMarcolli_ImprovedStandardModel}
A.~H. Chamsedinne, A.~Connes, and M.~Marcolli.
\newblock Gravity and the standard model with neutrino mixing.
\newblock {\em Advances in Theoretical and Mathematical Physics},
  11(6):991--1089, 2007.

\bibitem{connes1994noncommutative}
A.~Connes.
\newblock {\em Noncommutative geometry}.
\newblock Springer, 1994.

\bibitem{connes1995noncommutative}
A.~Connes.
\newblock Noncommutative geometry and reality.
\newblock {\em Journal of Mathematical Physics}, 36(11):6194--6231, 1995.

\bibitem{connes2000short}
A.~Connes.
\newblock A short survey of noncommutative geometry.
\newblock {\em Journal of Mathematical Physics}, 41(6):3832--3866, 2000.

\bibitem{connes2019noncommutative}
A.~Connes and M.~Marcolli.
\newblock {\em Noncommutative geometry, quantum fields and motives}, volume~55.
\newblock American Mathematical Soc., 2019.

\bibitem{d2022numerical}
M.~D'Arcangelo.
\newblock {\em Numerical simulation of random Dirac operators}.
\newblock PhD thesis, University of Nottingham, 2022.

\bibitem{deift1999orthogonal}
P.~Deift.
\newblock {\em Orthogonal polynomials and random matrices: a Riemann-Hilbert
  approach}, volume~3.
\newblock American Mathematical Soc., 1999.

\bibitem{dandrea2013metric}
F.~D’Andrea, F.~Lizzi, and J.~C. V{\'a}rilly.
\newblock Metric properties of the fuzzy sphere.
\newblock {\em Letters in Mathematical Physics}, 103:183--205, 2013.

\bibitem{eynardRandomMatrices}
B.~Eynard, T.~Kimura, and S.~Ribault.
\newblock {Random matrices}.
\newblock 138 pages, based on lectures by Bertrand Eynard at IPhT, Saclay, Jan.
  2016.

\bibitem{farnsworth2017productsoftriples}
S.~Farnsworth.
\newblock {The graded product of real spectral triples}.
\newblock {\em Journal of Mathematical Physics}, 58(2):023507, 02 2017.

\bibitem{Fathizadeh2019}
F.~Fathizadeh and M.~Khalkhali.
\newblock {\em Curvature in noncommutative geometry}, pages 321--420.
\newblock Springer International Publishing, Cham, 2019.

\bibitem{figalli2016universality}
A.~Figalli and A.~Guionnet.
\newblock Universality in several-matrix models via approximate transport maps.
\newblock {\em Acta Mathematica}, 217(1):115--159, 2016.

\bibitem{forrester2010log}
P.~J. Forrester.
\newblock {\em Log-gases and random matrices (LMS-34)}.
\newblock Princeton University Press, 2010.

\bibitem{Gesteau2024dhj}
E.~Gesteau and L.~Santilli.
\newblock {Explicit large $N$ von Neumann algebras from matrix models}.
\newblock 2 2024.

\bibitem{gilkey2018invariance}
P.~B. Gilkey.
\newblock {\em Invariance theory, the Heat Equation, and the Atiyah-Singer
  Index Theorem}.
\newblock Publish or Perish Inc., 1984.

\bibitem{glaser2017scaling}
L.~Glaser.
\newblock Scaling behaviour in random non-commutative geometries.
\newblock {\em Journal of Physics A: Mathematical and Theoretical},
  50(27):275201, 2017.

\bibitem{gracia2013elements}
J.~M. Gracia-Bond{\'\i}a, J.~C. V{\'a}rilly, and H.~Figueroa.
\newblock {\em Elements of noncommutative geometry}.
\newblock Springer Science \& Business Media, 2013.

\bibitem{guionnet2009large}
A.~Guionnet.
\newblock {\em Large random matrices}, volume~36.
\newblock Springer Science \& Business Media, 2009.

\bibitem{guionnet2005combinatorial}
A.~Guionnet and E.~Maurel-Segala.
\newblock Combinatorial aspects of matrix models.
\newblock {\em arXiv preprint math/0503064}, 2005.

\bibitem{gwynne2021random}
E.~Gwynne and J.~Miller.
\newblock Random walk on random planar maps: Spectral dimension, resistance and
  displacement.
\newblock 2021.

\bibitem{hessam2022bootstrapping}
H.~Hessam, M.~Khalkhali, and N.~Pagliaroli.
\newblock Bootstrapping dirac ensembles.
\newblock {\em Journal of Physics A: Mathematical and Theoretical},
  55(33):335204, 2022.

\bibitem{hessam2022double}
H.~Hessam, M.~Khalkhali, and N.~Pagliaroli.
\newblock Double scaling limits of dirac ensembles and liouville quantum
  gravity.
\newblock {\em arXiv preprint arXiv:2204.14206}, 2022.

\bibitem{hessam2022noncommutative}
H.~Hessam, M.~Khalkhali, N.~Pagliaroli, and L.~S. Verhoeven.
\newblock From noncommutative geometry to random matrix theory.
\newblock {\em Journal of Physics A: Mathematical and Theoretical},
  55(41):413002, 2022.

\bibitem{johansson1998fluctuations}
K.~Johansson.
\newblock On fluctuations of eigenvalues of random hermitian matrices.
\newblock 1998.

\bibitem{kazakov1999two}
V.~A. Kazakov and P.~Zinn-Justin.
\newblock Two-matrix model with abab interaction.
\newblock {\em Nuclear Physics B}, 546(3):647--668, 1999.

\bibitem{khalkhali2020phase}
M.~Khalkhali and N.~Pagliaroli.
\newblock Phase transition in random noncommutative geometries.
\newblock {\em Journal of Physics A: Mathematical and Theoretical},
  54(3):035202, 2020.

\bibitem{khalkhali2022spectral}
M.~Khalkhali and N.~Pagliaroli.
\newblock Spectral statistics of dirac ensembles.
\newblock {\em Journal of Mathematical Physics}, 63(5):053504, 2022.

\bibitem{khalkhali2023coloured}
M.~Khalkhali and N.~Pagliaroli.
\newblock Coloured combinatorial maps and quartic bi-tracial 2-matrix ensembles
  from noncommutative geometry.
\newblock {\em arXiv preprint arXiv:2312.10530}, 2023.

\bibitem{krajewski1998classification}
T.~Krajewski.
\newblock Classification of finite spectral triples.
\newblock {\em Journal of Geometry and Physics}, 28(1-2):1--30, 1998.

\bibitem{madore1992fuzzy}
J.~Madore.
\newblock The fuzzy sphere.
\newblock {\em Classical and Quantum Gravity}, 9(1):69, 1992.

\bibitem{perez2021multimatrix}
C.~I. P{\'e}rez-S{\'a}nchez.
\newblock On multimatrix models motivated by random noncommutative geometry i:
  the functional renormalization group as a flow in the free algebra.
\newblock In {\em Annales Henri Poincar{\'e}}, volume~22, pages 3095--3148.
  Springer, 2021.

\bibitem{perez2022computing}
C.~I. P{\'e}rez-S{\'a}nchez.
\newblock Computing the spectral action for fuzzy geometries: from random
  noncommutative geometry to bi-tracial multimatrix models.
\newblock {\em Journal of Noncommutative Geometry}, 2022.

\bibitem{perez2022multimatrix}
C.~I. Perez-Sanchez.
\newblock On multimatrix models motivated by random noncommutative geometry ii:
  A yang-mills-higgs matrix model.
\newblock In {\em Annales Henri Poincar{\'e}}, volume~23, pages 1979--2023.
  Springer, 2022.

\bibitem{ramirez2011beta}
J.~Ramirez, B.~Rider, and B.~Vir{\'a}g.
\newblock Beta ensembles, stochastic airy spectrum, and a diffusion.
\newblock {\em Journal of the American Mathematical Society}, 24(4):919--944,
  2011.

\bibitem{schmidt2017square}
M.~D. Schmidt.
\newblock Square series generating function transformations.
\newblock {\em Journal of Inequalities \& Special Functions}, 8(2), 2017.

\bibitem{van2022one}
T.~D. van Nuland and W.~D. van Suijlekom.
\newblock One-loop corrections to the spectral action.
\newblock {\em Journal of High Energy Physics}, 2022(5):1--15, 2022.

\bibitem{van2015noncommutative}
W.~D. Van~Suijlekom.
\newblock {\em Noncommutative geometry and particle physics}.
\newblock Springer, 2015.

\bibitem{verhoeven2023thesis}
L.~S. Verhoeven.
\newblock {\em Geometry in spectral triples: immersions and fermionic fuzzy
  geometries}.
\newblock PhD thesis, The University of Western Ontario, 2023.

\bibitem{wulkenhaar2019quantum}
R.~Wulkenhaar.
\newblock Quantum field theory on noncommutative spaces.
\newblock {\em Advances in Noncommutative Geometry: On the Occasion of Alain
  Connes' 70th Birthday}, pages 607--690, 2019.

\bibitem{zwillinger2007table}
D.~Zwillinger and A.~Jeffrey.
\newblock {\em Table of integrals, series, and products}.
\newblock Elsevier, 2007.

\end{thebibliography}

\end{document}